\RequirePackage{fix-cm}

\documentclass[a4paper]{myjournal}

\usepackage[utf8]{inputenc}
\usepackage[english]{babel}

\usepackage{amsmath,amsfonts,amssymb,amsthm}
\usepackage{array}
\usepackage{tabularx}
\usepackage{url}
\usepackage{placeins}

\usepackage[algo2e,ruled,vlined,linesnumbered]{algorithm2e}
\SetCommentSty{textit}
\DontPrintSemicolon
\IncMargin{-\parindent}
\SetAlCapHSkip{0pt}
\SetAlgoLined

\SetKwProg{Function}{Function}{}{}
\SetKwFunction{LCPCompare}{LCP-Compare}

\usepackage{versions}
\excludeversion{FORABSTRACT}
\includeversion{FORREPORT}
\excludeversion{FORJOURNAL}

\usepackage{tikz,pgfplots}
\usetikzlibrary{positioning}

\tikzset{
  baseline=(current bounding box.center),
  HTline/.style={black!50,dashed},
}

\usetikzlibrary{external}
\tikzset{
  external/prefix={tikz/},
  external/optimize=true,
  external/mode={list and make},
}

\pgfplotscreateplotcyclelist{mycolor1}{%
  red, every mark/.append style={solid,scale=0.8}, mark=o \\%
  blue, every mark/.append style={solid,scale=0.7}, mark=square \\%
  violet, every mark/.append style={solid,scale=0.95}, mark=x \\%
  green!70!black, every mark/.append style={solid,scale=0.95}, mark=diamond \\%
  orange!40!red, every mark/.append style={solid,scale=0.95}, mark=x, dashed \\%
  teal, every mark/.append style={solid,scale=0.95}, mark=asterisk, dashed \\%
  brown!70!black, every mark/.append style={solid,scale=0.95}, mark=triangle, dashed \\%
}

\pgfplotscreateplotcyclelist{mycolor2}{%
  red, every mark/.append style={solid,scale=0.8}, mark=o \\%
  blue, every mark/.append style={solid,scale=0.7}, mark=square \\%
  green!70!black, every mark/.append style={solid,scale=0.95}, mark=diamond \\%
  orange!40!red, every mark/.append style={solid,scale=0.95}, mark=x, dashed \\%
  teal, every mark/.append style={solid,scale=0.95}, mark=asterisk, dashed \\%
  brown!70!black, every mark/.append style={solid,scale=0.95}, mark=triangle, dashed \\%
  purple!80!blue, every mark/.append style={solid,scale=0.7}, mark=square, densely dotted \\%
  olive, every mark/.append style={solid,scale=0.8}, mark=o, densely dotted \\%
  gray, every mark/.append style={solid,scale=0.95}, mark=x, densely dotted \\%
}

\pgfplotsset{
  major grid style={thin,dotted,color=black!50},
  minor grid style={thin,dotted,color=black!50},
  grid,
  every axis/.append style={
    line width=0.5pt,
    tick style={
      line cap=round,
      thin,
      major tick length=4pt,
      minor tick length=2pt,
    },
  },
  legend cell align=left,
  legend style={
    /tikz/every even column/.append style={column sep=3mm,black},
    /tikz/every odd column/.append style={black},
  },
  plotSpeedupMini/.style={
    width=67mm,height=49mm,
    xlabel near ticks,
    ylabel absolute=true,
    every axis y label/.append style={yshift=-13pt},
    max space between ticks=18pt,
    title style={yshift=-3pt},
  },
  plotSpeedup64/.style={
    plotSpeedupMini,xtick={1,8,16,32,48,64},
    cycle list name={mycolor1},
  },
  plotSpeedup48/.style={
    plotSpeedupMini,xtick={1,6,12,24,36,48},
    cycle list name={mycolor1},
  },
  plotSpeedup16/.style={
    plotSpeedupMini,xtick={1,2,4,6,8,12,16},
    cycle list name={mycolor1},
  },
  plotSpeedup8/.style={
    plotSpeedupMini,xtick={1,...,8},ytick={0,...,5},
    cycle list name={mycolor2},
  },
}

\usepackage{hyperref}
\hypersetup{
  breaklinks=true,
  colorlinks=true,
  citecolor=blue,
  linkcolor=blue,
  urlcolor=blue,
  bookmarksnumbered,
  bookmarksopen,
  pdftitle={Engineering Parallel String Sorting},
  pdfauthor={Timo Bingmann, Andreas Eberle, Peter Sanders},
  pdfsubject={parallel string sorting algorithms, super scalar string sample sort, multiway LCP-mergesort},
  pdfkeywords={parallel string sorting, parallel algorithm, string sorter, sample sort, mergesort, lcp},
}

\hypersetup{pdfpagescrop={10 130 485 830}}

\usepackage{breakurl}

\usepackage{amsfonts}

\newcommand{\ceil}[1]{\left\lceil #1\right\rceil}

\newcommand{\set}[1]{\left\{ #1\right\}}
\newcommand{\gilt}{:}

\newcommand{\realrange}[2]{\left[#1, #2\right]}

\newcommand{\unitrange}[2]{\realrange{0}{1}}

\newcommand{\Oh}[1]{\mathcal{O}\!\left( #1\right)}
\newcommand{\oh}[1]{\mathrm{o}\!\left( #1\right)}
\newcommand{\Th}[1]{\Theta\!\left( #1\right)}

\newcommand{\llabel}[1]{\label{\labelprefix:#1}}
\newcommand{\labelprefix}{} 

\newcommand{\discussionsize}{\small}

\newcommand{\punkt}{\enspace .}

\newenvironment{code}{\noindent
\begin{tabbing}%
\hspace{2em}\=\hspace{2em}\=\hspace{2em}\=\hspace{2em}\=\hspace{2em}\=%
\hspace{2em}\=\hspace{2em}\=\hspace{2em}\=\hspace{2em}\=\hspace{2em}\=%
\kill}{\end{tabbing}}

\newcommand{\labelcommand}{}
\newcommand{\captiontext}{}
\newsavebox{\codeparam}
\newcounter{lineNumber}
\newenvironment{disscodepos}[3]{%
\renewcommand{\labelcommand}{#2}%
\renewcommand{\captiontext}{#3}%
\sbox{\codeparam}{\parbox{\textwidth}{#3}}%
\begin{figure}[#1]\begin{center}\begin{code}\setcounter{lineNumber}{1}}{%
\end{code}\end{center}\caption{\llabel{\labelcommand}\captiontext}\end{figure}}

{\end{disscodepos}}

\newdimen\endofsize\endofsize=0.5em
\def\endofbeweis{~\quad\hglue\hsize minus\hsize
                 \hbox{\vrule height \endofsize width
\endofsize}\par}

\newcommand{\lcp}{\mathrm{lcp}}
\newcommand{\Strings}{\mathcal{S}}

\def\Oh#1{\mathcal{O}(#1)}
\def\oh#1{\mathrm{o}(#1)}

\newcommand{\arr}[1]{[\, #1 \,]}

\newcommand{\Inc}{{++}}
\newcommand{\Dec}{{--}}

\newcommand{\Rem}[1]{\tcp*{#1}}
\newcommand{\Remi}[1]{\tcp*[f]{#1}}

\begin{document}

\title{Engineering Parallel String Sorting}

\author{Timo Bingmann \and Andreas Eberle \and Peter Sanders}

\institute{
Karlsruhe Institute of Technology, Karlsruhe, Germany\\
\email{\{bingmann,sanders\}@kit.edu}}



\maketitle

\begin{abstract}
  We discuss how string sorting algorithms can be parallelized on modern
  multi-core shared memory machines.  As a synthesis of the best sequential
  string sorting algorithms and successful parallel sorting algorithms for
  atomic objects, we first propose string sample sort. The algorithm makes
  effective use of the memory hierarchy, uses additional word level parallelism,
  and largely avoids branch mispredictions. Then we focus on NUMA architectures,
  and develop parallel multiway LCP-merge and -mergesort to reduce the number of
  random memory accesses to remote nodes.  Additionally, we parallelize variants
  of multikey quicksort and radix sort that are also useful in certain
  situations. Comprehensive experiments on five current multi-core platforms are
  then reported and discussed. The experiments show that our implementations
  scale very well on real-world inputs and modern machines.
\end{abstract}

\section{Introduction}

Sorting is perhaps the most studied algorithmic problem in computer science.
While the most simple model for sorting assumes \emph{atomic} keys, an important
class of keys are strings or vectors to be sorted lexicographically. Here, it is important
to exploit the structure of the keys to avoid costly repeated operations on the
entire string.  String sorting is for example needed in database index
construction, some suffix sorting algorithms, or MapReduce tools. Although there
is a correspondingly large volume of work on sequential string sorting, there is
very little work on parallel string sorting. This is surprising since
parallelism is now the only way to get performance out of Moore's law so that
any performance critical algorithm needs to be parallelized. We therefore
started to look for practical parallel string sorting algorithms for modern
multi-core shared memory machines. Our focus is on large inputs which fit into
RAM.  This means that besides parallelization we have to take the memory
hierarchy, layout, and processor features like the high cost of branch
mispredictions, word parallelism, and super scalar processing into account.
Looking beyond single-socket multi-core architectures, we also consider
many-core machines with multiple sockets and non-uniform memory access (NUMA).

In Section~\ref{sec:basic-sequential} we give an overview of basic sequential
string sorting algorithms, acceleration techniques and more related work. We
then propose our first new string sorting algorithm, super scalar string sample
sort (S$^5$), in Section~\ref{sec:s5}. Thereafter, we turn our focus to NUMA
architectures in Section~\ref{sec:para-mergesort}, and develop parallel
LCP-aware multiway merging as a top-level algorithm for combining presorted
sequences. Broadly speaking, we propose both multiway distribution-based string
sorting with S$^5$ and multiway merge-based string sorting with LCP-aware
mergesort, and parallelize both approaches.

Section~\ref{sec:more-parasort} describes parallelizations of caching multikey
quicksort and radix sort, which are two more competitors. We then compare both
parallel and sequential string sorting algorithms experimentally in
Section~\ref{sec:experiments}.

For all our input instances, except random strings, parallel S$^5$ achieves
higher speedups on modern single-socket multi-core machines than our own
parallel multikey quicksort and radixsort implementations, which are already
better than any previous ones. For multi-socket NUMA machines, parallel multiway
LCP-merge with node-local parallel S$^5$ achieves higher speedups for large
real-world inputs than all other implementations in our experiment.

Shorter versions of Section~\ref{sec:s5}, \ref{sec:more-parasort} and
\ref{sec:experiments} have appeared in our conference
paper~\cite{bingmann2013parallel}.  We would like to thank our students Florian
Drews, Michael Hamann, Christian Käser, and Sascha Denis Knöpfle who implemented
prototypes of our ideas.

\section{Preliminaries}\label{sec:prelim}

Our input is a set $\Strings = \{ s_1,\ldots,s_n \}$ of $n$ strings with total
length $N$.  A \emph{string} $s$ is a one-based array of $|s|$ characters from
the \emph{alphabet} $\Sigma = \{ 1,\ldots,\sigma \}$.  We assume the canonical
lexicographic ordering relation `$<$' on strings, and our goal is to sort
$\Strings$ lexicographically.  For the implementation and pseudo-code, we
require that strings are zero-terminated, i.e. $s[|s|] = 0 \notin \Sigma$, but
this convention can be replaced using other end-of-string indicators, like
string length.

Let $D$ denote the \emph{distinguishing prefix size} of $\Strings$, i.e., the
total number of characters that need to be inspected in order to establish the
lexicographic ordering of $\Strings$. $D$ is a natural lower bound for the
execution time of sequential string sorting. If, moreover, sorting is based on
character comparisons, we get a lower bound of $\Omega(D + n \log n)$.

Sets of strings are usually represented as arrays of pointers to the beginning
of each string. Note that this indirection means that, in general, every access
to a string incurs a cache fault even if we are scanning an array of strings.
This is a major difference to atomic sorting algorithms where scanning is very
cache efficient.  Our target machine is a shared memory system supporting $p$
hardware threads or processing elements (PEs), on $\Th{p}$ cores.

\subsection{Notation and Pseudo-code}

The algorithms in this paper are written in a pseudo-code language, which mixes
Pascal-like control flow with array manipulation and mathematical set notation.
This enables powerful expressions like $A := \arr{ (i^2 \bmod 7, i) \mid i \in
  [0 \mathop{:} 5) }$, which sets $A$ to be the array of pairs $\arr{ (0,0), (1,1), (4,2),
  (2,3), (2,4) }$. We write ordered sequences like arrays using square brackets
$\arr{ \ldots }$, overload `$+$' to also concatenate arrays, and let $[1 \mathop{:} n] :=
\arr{ 1,\ldots,n }$ and $[1 \mathop{:} n) := \arr{ 1,\ldots,n-1 }$ be ranges of integers.  To
make array operations more concise, we assume $A_i$ and $A[i]$ both to be the
$i$-th element in the array $A$. We do not allocate or declare arrays and
variables beforehand, so $A_i := 1$ also implicitly defines an array $A$. The
unary operators `$\Inc$` and `$\Dec$` increment and decrement integer variables
by one.

To avoid special cases, we use the following sentinels: `$\varepsilon$' is the
empty string, which is lexicographically smaller than any other string,
`$\infty$' is a character or string larger than any other character or string,
and `$\bot$' is an undefined variable.

For two arrays $s$ and $t$, let $\lcp(s,t)$ denote the length of the
\emph{longest common prefix} (LCP) of $s$ and $t$. This function is symmetric,
and for one-based arrays the LCP value denotes the last index where $s$ and $t$
match, while position $\lcp(s,t)+1$ differs in $s$ and $t$, if it exists.  In a
sequence $x$ let $\lcp_x(i)$ denote $\lcp(x_{i-1},x_i)$. For a sorted sequence
of strings $\Strings = \arr{ s_1,\ldots,s_n }$ the \emph{associated LCP~array}
$H$ is $\arr{ \bot,h_2,\ldots,h_n }$ with $h_i = \lcp_{\Strings}(i) =
\lcp(s_{i-1},s_i)$. For the empty string $\varepsilon$, let $\lcp(\varepsilon,s)
= 0$ for any string $s$.

We will often need the sum over all items in an LCP array $H$ (excluding the
first), and denote this as $L(H) := \sum_{i=2}^n H_i$, or just $L$ if $H$ is
clear from the context. The distinguishing prefix size $D$ and $L$ are related
but not identical. While $D$ includes all characters counted in $L$,
additionally, $D$ also accounts for the distinguishing characters, some string
terminators and characters of the first string. In general, we have $D \geq L$.

\section{Basic Sequential String Sorting Algorithms}\label{sec:basic-sequential}

We begin by giving an overview of most efficient sequential string sorting
algorithms. Nearly all algorithms classify the original string set $\Strings$
into smaller sets with a distinct common prefix. The smaller sets are then
sorted further recursively, until the sets contain only one item or another
string sorter is called.

\emph{Multikey quicksort} \cite{bentley1997fast} is a simple but effective
adaptation of quicksort to strings (called multi-key data).  When all strings in
$\Strings$ have a common prefix of length $\ell$, the algorithm uses character
$c=s[\ell+1]$ of a pivot string $s\in \Strings$ (e.g. a pseudo-median) as a
\emph{splitter} character. $\Strings$ is then partitioned into $\Strings_<$,
$\Strings_=$, and $\Strings_>$ depending on comparisons of the $(\ell+1)$-th
character with $c$. Recursion is done on all three subproblems. The key
observation is that the strings in $\Strings_=$ have common prefix length
$\ell+1$ which means that compared characters found to be equal with $c$ will
never be considered again. Insertion sort is used as a base case for constant
size inputs. This leads to a total execution time of $\Oh{D+n\log n}$. Multikey
quicksort works well in practice in particular for inputs which fit into the
cache. Since a variant of multikey quicksort was the overall best sequential
algorithm in our experiments, we develop a parallel version in
Section~\ref{sec:para-mkqs}.

\emph{MSD radix sort}
\cite{mcilroy1993engineering,ng2007cache,karkkainen2009engineering} with common
prefix length $\ell$ looks at the $(\ell+1)$-th character producing $\sigma$
subproblems which are then sorted recursively with common prefix $\ell+1$. This
is a good algorithm for large inputs and small alphabets since it uses the
maximum amount of information within a single character. For input sizes
$\oh{\sigma}$ MSD radix sort is no longer efficient and one has to switch to a
different algorithm for the base case. The running time is $\Oh{D}$ plus the
time for solving the base cases. Using multikey quicksort for the base case
yields an algorithm with running time $\Oh{D+n\log\sigma}$. A problem with large
alphabets is that one will get many cache faults if the cache cannot support
$\sigma$ concurrent output streams (see \cite{mehlhorn2003scanning} for
details). We discuss parallel radix sorting in Section~\ref{sec:para-radixsort}.

\emph{Burstsort} dynamically builds a trie data structure for the input
strings. In order to reduce the involved work and to become cache efficient, the
trie is build lazily -- only when the number of strings referenced in a
particular subtree of the trie exceeds a threshold, this part is expanded. Once
all strings are inserted, the relatively small sets of strings stored at the
leaves of the trie are sorted recursively (for more details refer to
\cite{sinha2004cache-conscious,sinha2007cache-efficient,sinha2010engineering}
and the references therein).

\emph{LCP-Mergesort} is an adaptation of mergesort to strings that saves and
reuses the LCPs of consecutive strings in the sorted subproblems
\cite{ng2008merging}. In section~\ref{sec:para-mergesort}, we develop a parallel
multiway variant of LCP-merge, which is used to improve performance on NUMA
machines. Our multiway LCP-merge is also interesting for merging of
string sets stored in external memory.

\emph{Insertion sort} \cite{knuth1998sorting} keeps an ordered array, into which
unsorted items are inserted by linearly scanning for their correct position. If
strings are considered atomic, then full string comparisons are done during the
linear scan. This is particularly cache-efficient and the algorithm is commonly
used as base case sorter.  However, if one keeps additionally the associated LCP
array, the number of character comparisons can be decreased, trading them for
integer comparisons of LCPs. We needed a base case sorter that also calculates
the LCP array and found no reference for LCP-aware insertion sort in the
literature, so we describe the algorithm in Section~\ref{sec:lcp-inssort}.

\subsection{Architecture Specific Enhancements}

To increase the performance of basic sequential string sorting algorithms on
real hardware, we have to take its architecture into consideration.  In the
following list we highlight some of most important optimization principles.

\emph{Memory access time} varies greatly in modern systems.  While the RAM model
considers all memory accesses to take unit time, current architectures have
multiple levels of cache, require additional memory access on TLB misses, and
may have to request data from ``remote'' nodes on NUMA systems.  While there are
few hard guarantees, we can still expect recently used memory to be in cache and
use these assumptions to design cache-efficient algorithms.  Furthermore, on
NUMA systems we can instruct the kernel on how to distribute memory by
specifying allocation policies for memory segments

\emph{Caching of characters} is very important for modern memory hierarchies as
it reduces the number of cache misses due to random access on strings.  When
performing character lookups, a caching algorithm copies successive characters
of the string into a more convenient memory area.  Subsequent sorting steps can
then avoid random access, until the cache needs to be refilled.  This technique
has successfully been applied to radix sort \cite{ng2007cache}, multikey
quicksort \cite{rantala2007web}, and in its extreme to burstsort
\cite{sinha2007cache-efficient}. However, caching comes at the cost of increased
space requirements and memory accesses, hence a good trade-off must be found.

\emph{Super-Alphabets} can be used to accelerate string sorting algorithms which
originally look only at single characters.  Instead, multiple characters are
grouped as one and sorted together.  However, most algorithms are very sensitive
to large alphabets, thus the group size must be chosen carefully.  This approach
results in 16-bit MSD radix sort and fast sorters for DNA strings.  If the
grouping is done to fit many characters into a machine word for processing as a
whole block using arithmetic instructions, then this is also called \emph{word
  parallelism}.

\emph{Unrolling, fission and vectorization of loops} are methods to exploit
out-of-order execution and super scalar parallelism now standard in modern CPUs.
The processor's instruction scheduler automatically analyses the machine code,
detects data dependencies and can dispatch multiple parallel operations.
However, only specific, simple data independencies can be detected and thus
inner loops must be designed with care (e.g. for radix sort
\cite{karkkainen2009engineering}). The performance increase by reorganizing
loops is most difficult to predict.

\subsection{More Related Work}

There is a huge amount of work on parallel sorting of atomic objects so that we
can only discuss the most relevant results. Besides (multiway) mergesort,
perhaps the most practical parallel sorting algorithms are parallelizations of
radix sort (e.g.~\cite{wassenberg2011engineering}) and
quicksort~\cite{tsigas2003simple} as well as sample
sort~\cite{blelloch1991comparison}.

There is some work on PRAM algorithms for string sorting
(e.g. \cite{hagerup94optimal}). By combining pairs of adjacent characters into
single characters, one obtains algorithms with work $\Oh{N\log N}$ and time
$\Oh{\log N/\log\log N}$. Compared to the sequential algorithms this is
suboptimal unless $D=\Oh{N}=\Oh{n}$ and with this approach it is unclear how to
avoid work on characters outside distinguishing prefixes.

We found no publications on practical parallel string sorting, aside from our
conference paper~\cite{bingmann2013parallel}. However, Ta\-kuya Akiba has
implemented a parallel radix sort \cite{akiba2011radixsort}, Tommi Rantala's
library~\cite{rantala2007web} contains multiple parallel mergesorts and a
parallel SIMD variant of multikey quicksort, and Nagaraja
Shamsundar~\cite{shamsundar2009lcpmergesort} also parallelized Waihong Ng's
LCP-mergesort \cite{ng2008merging}. Of all these implementations, only the radix
sort by Akiba scales reasonably well to many-core architectures.  We discuss the
scalability issues of these implementations in Section~\ref{sec:exp-parallel}.

\section{Super Scalar String Sample Sort (\texorpdfstring{S$^5$}{S5})}\label{sec:s5}

Already in a sequential setting, theoretical considerations and experiments (see
Section~\ref{sec:exp-sequential}) indicate that \emph{the} best string sorting
algorithm does not exist.  Rather, it depends at least on $n$, $D$, $\sigma$,
and the hardware.  Therefore we decided to parallelize several algorithms taking
care that components like data distribution, load balancing or base case sorter
can be reused.  Remarkably, most algorithms in
Section~\ref{sec:basic-sequential} can be parallelized rather easily and we will
discuss parallel versions in
Sections~\ref{sec:s5}--\ref{sec:more-parasort}. However, none of these
parallelizations make use of the striking new feature of modern many-core
systems: many multi-core processors with individual cache levels but relatively
few and slow memory channels to shared RAM. Therefore we decided to design a new
string sorting algorithm based on \emph{sample sort}
\cite{frazer1970samplesort}, which exploits these properties.  Preliminary
result on string sample sort have been reported in the bachelor thesis of Sascha
Denis Knöpfle~\cite{knoepfle2012string}.

\subsection{Traditional (Parallel) Atomic Sample Sort}\label{sec:ss-atomic}

Sample sort \cite{frazer1970samplesort,blelloch1991comparison} is a generalization of quicksort working
with $k-1$ pivots at the same time.  For small inputs, sample sort uses some
sequential base case sorter.  Larger inputs are split into $k$ \emph{buckets}
$b_1,\ldots,b_k$ by determining $k-1$ splitter keys $x_1\leq \cdots\leq x_{k-1}$
and then classifying the input elements -- element $s$ goes to bucket $b_i$ if
$x_{i-1}< s \leq x_i$ (where $x_0$ and $x_k$ are defined as sentinel elements --
$x_0$ being smaller than all possible input elements and $x_k$ being larger).
Splitters can be determined by drawing a random sample of size $\alpha k-1$ from
the input, sorting it, and then taking every $\alpha$-th element as a
splitter. Parameter $\alpha$ is the \emph{oversampling} factor. The buckets are
then sorted recursively and concatenated. ``Traditional'' parallel sample sort
chooses $k=p$ and uses a sample big enough to assure that all buckets have
approximately equal size.  Sample sort is also attractive as a sequential
algorithm since it is more cache efficient than quicksort and since it is
particularly easy to avoid branch mispredictions (super scalar sample sort --
S$^4$) \cite{sanders2004super}. In this case, $k$ is chosen in such a way that
classification and data distribution can be done in a cache efficient way.

\subsection{String Sample Sort}\label{sec:ss-string}

In order to adapt the atomic sample sort from the previous section to strings,
we have to devise an efficient classification algorithm.  Most importantly, we
want to avoid comparing whole strings needlessly, and thus focus on character
comparisons.  Also, in order to approach total work $\Oh{D+n\log n}$, we have to
use the information gained during classification in the recursive calls. This
can be done by observing that strings in buckets have a common prefix depending
on the LCP of the two splitters:
\begin{equation}\label{eq:lcp}
\forall 1\leq i \leq k\gilt \forall s,t\in b_i\gilt \lcp(s,t)\geq \lcp_x(i)\punkt
\end{equation}
Another issue is that we have to reconcile the parallelization and load
balancing perspective from traditional parallel sample sort with the cache
efficiency perspective of super scalar sample sort, where the splitters are
designed to fit into cache. We do this by using dynamic load balancing which
includes parallel execution of recursive calls as in parallel quicksort. Dynamic
load balancing is very important and probably unavoidable for parallel string
sorting, because any algorithm must adapt to the input string set's
characteristics.

\subsection{Super Scalar String Sample Sort (\texorpdfstring{S$^{\,5}$}{S5}) -- A Pragmatic Solution}

We adapt the implicit binary search tree approach used in (atomic) super scalar
sample sort (S$^4$)~\cite{sanders2004super} to strings.  Algorithm~\ref{alg:s5}
shows pseudo-code of one variant of S$^5$ as a guideline through the following
discussion, and Figure~\ref{fig:ternary-tree} illustrates the classification
tree with buckets and splitters.

\begin{figure}[t]\centering
  \begin{tikzpicture}[
  xscale=0.6,
  bucket/.style={draw,rectangle,inner sep=2pt},
  keynode/.style={draw,circle,inner sep=1pt},
  equall/.style={fill=white,inner sep=2pt,text depth=-1pt},
  ]

\node (k1) [keynode] at (7,3) {\vphantom{0}$x_3$};

\node (k2) [keynode] at (3,2) {\vphantom{0}$x_1$};
\node (k3) [keynode] at (11,2) {\vphantom{0}$x_5$};

\node (k4) [keynode] at (1,1) {\vphantom{0}$x_0$};
\node (k5) [keynode] at (5,1) {\vphantom{0}$x_2$};
\node (k6) [keynode] at (9,1) {\vphantom{0}$x_4$};
\node (k7) [keynode] at (13,1) {\vphantom{0}$x_6$};

\node (b0)  [bucket] at (0,0) {\vphantom{0}$b_0$};
\node (b1)  [bucket] at (1,0) {\vphantom{0}$b_1$};
\node (b2)  [bucket] at (2,0) {\vphantom{0}$b_2$};
\node (b3)  [bucket] at (3,0) {\vphantom{0}$b_3$};
\node (b4)  [bucket] at (4,0) {\vphantom{0}$b_4$};
\node (b5)  [bucket] at (5,0) {\vphantom{0}$b_5$};
\node (b6)  [bucket] at (6,0) {\vphantom{0}$b_6$};
\node (b7)  [bucket] at (7,0) {\vphantom{0}$b_7$};
\node (b8)  [bucket] at (8,0) {\vphantom{0}$b_8$};
\node (b9)  [bucket] at (9,0) {\vphantom{0}$b_9$};
\node (b10) [bucket] at (10,0) {\vphantom{0}$b_{10}$};
\node (b11) [bucket] at (11,0) {\vphantom{0}$b_{11}$};
\node (b12) [bucket] at (12,0) {\vphantom{0}$b_{12}$};
\node (b13) [bucket] at (13,0) {\vphantom{0}$b_{13}$};
\node (b14) [bucket] at (14,0) {\vphantom{0}$b_{14}$};

\draw (k1) -- node[above] {$<$} (k2);
\draw (k1) -- node[equall] {$=$} (b7);
\draw (k1) -- node[above] {$>$} (k3);

\draw (k2) -- node[above] {$<$} (k4);
\draw (k2) -- node[equall] {$=$} (b3);
\draw (k2) -- node[above] {$>$} (k5);

\draw (k3) -- node[above] {$<$} (k6);
\draw (k3) -- node[equall] {$=$} (b11);
\draw (k3) -- node[above] {$>$} (k7);

\draw (k4) -- node[left] {$<$} (b0);
\draw (k4) -- node[equall] {$=$} (b1);
\draw (k4) -- node[right] {$>$} (b2);

\draw (k5) -- node[left] {$<$} (b4);
\draw (k5) -- node[equall] {$=$} (b5);
\draw (k5) -- node[right] {$>$} (b6);

\draw (k6) -- node[left] {$<$} (b8);
\draw (k6) -- node[equall] {$=$} (b9);
\draw (k6) -- node[right] {$>$} (b10);

\draw (k7) -- node[left] {$<$} (b12);
\draw (k7) -- node[equall] {$=$} (b13);
\draw (k7) -- node[right] {$>$} (b14);

\end{tikzpicture}
  \caption{Ternary classification tree for $v = 7$ splitters and $k = 15$ buckets.}\label{fig:ternary-tree}
\end{figure}
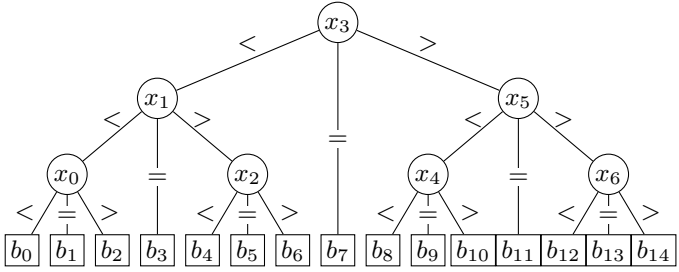

Rather than using whole strings as arbitrarily long splitters, or all characters
of the alphabet as in radix sort, we design the splitter keys to consist of
\emph{as many characters as fit into a machine word}.  In the following let $w$
denote the number of characters fitting into one machine word (for 8-bit
characters and 64-bit machine words we would have $w=8$).  We choose $v=2^d-1$
splitters $x_0,\ldots,x_{v-1}$ (for some integer $d$) from a sorted sample to
construct a \emph{perfect binary search tree}, which is used to classify a set
of strings based on the next $w$ characters at common prefix $h$. The main
disadvantage of this approach is that we may have many input strings whose next
$w$ characters are identical.  For these strings, the classification does not
reveal much information. We make the best out of such inputs by explicitly
defining \emph{equality buckets} for strings whose next $w$ characters exactly
match $x_i$.  For equality buckets, we can increase the common prefix length by
$w$ in the recursive calls, i.e., these characters will never be inspected
again.  In total, we have $k = 2 v + 1$ different buckets $b_0,\ldots,b_{2v}$
for a ternary search tree (see Figure~\ref{fig:ternary-tree}).

\begin{algorithm2e}[t]
\caption{Sequential Super Scalar String Sample Sort -- a single step}\label{alg:s5}\normalsize

\KwIn{$\Strings = \{ s_1,\ldots,s_n \}$ a set of strings with common prefix $h$.}

$p_i := \FuncSty{chars}_h(s_{\FuncSty{random}(1,\ldots,n)}) \quad\forall\, i = 1,\ldots,v \alpha + \alpha \!-\! 1 $ \Rem{Read sample $p$ of $\Strings$,}
$\FuncSty{sort}(\arr{ p_1,\ldots,p_{v \alpha + \alpha - 1} })$ \Rem{sort it, and select}
$\arr{ x_0,x_1,x_2,\ldots,x_{v-1},x_v } := \arr{ p_\alpha,p_{2\alpha},p_{3\alpha},\ldots,p_{v \alpha},\infty }$ \Rem{equidistant splitters.}\nllabel{alg:s5:sample}
$t := \arr{ x_{\frac{v'}{2}}, x_{\frac{v'}{4}}, x_{\frac{3v'}{4}}, x_{\frac{v'}{8}}, x_{\frac{3v'}{8}}, x_{\frac{5v'}{8}}, x_{\frac{7v'}{8}}, \ldots }$ \KwSty{with} $v' := v + 1$ \Rem{Construct tree,}\nllabel{alg:s5:tree}
$h' := \arr{ 0 } + \arr{ \lcp(x_{i-1},x_i) \mid i = 1,\ldots,v-1 } + \arr{ 0 }$ \Rem{and save LCPs of splitters.}\nllabel{alg:s5:slcp}
\For(\Remi{Process strings (interleavable loop).}\nllabel{alg:s5:strloop})
{$j := 1,\ldots,n$}
{
  \KwSty{local} $i := 1$,\quad $c := \FuncSty{chars}_h(s_j)$ \Rem{Start at root, get $w$ chars from $s_j$,}
  \For(\Remi{and traverse tree (unrollable loop)}\nllabel{alg:s5:treeloop})
  {$1,\ldots,\log_2(v+1)$}
  {
    $i := 2 i + (c \leq t_i)$ \Rem{without branches using ``$(c \leq t_i)$'' $\in \set{0,1}$.}\nllabel{alg:s5:chartest}
  }
  $i := i - v+1$,\quad \KwSty{local} $m := 2 i$ \Rem{Calculate matching non-equality bucket.}
  \lIf(\Remi{Test for equality with next splitter.})
  {$x_i = c$\nllabel{alg:s5:eqtest}}
  {
    $m \Inc$
  }
  $o_j := m$ \Rem{Save final bucket number for string $s_j$ in an oracle.}\nllabel{alg:s5:end-strloop}
}
$b_i := 0 \quad\forall\, i = 0,\ldots,2v$ \nllabel{alg:s5:bktzero}\Rem{Inclusive prefix sum}
\lFor(\Remi{over bucket sizes})
{$i := 1,\ldots,n$}
{
  $(b_{o_i}) \Inc$\nllabel{alg:s5:count}
}
$\arr{ b_0,\ldots,b_{2v},b_{2v+1} } := \arr{ \sum_{j \leq i} b_j \mid i = 0,\ldots,2v } + \arr{n}$ \Rem{as fissioned loops.}\nllabel{alg:s5:prefixsum}
\lFor(\Remi{Reorder strings into new subsets.})
{$i := 1,\ldots,n$}
{
  $s'_{(b_{o_i})\Dec} := s_i$\nllabel{alg:s5:redistribute}
}

\KwOut{$\Strings'_i = \{ s'_j \mid j = b_i,\ldots,b_{i+1}-1 \text{ if } b_i <
  b_{i+1} \}$ are $i = 0,\ldots,2v$ string subsets with $\Strings'_i <
  \Strings'_{i+1}$. The subsets have common prefix $h + h'_i$ for $i$ even, and
  common prefix $h + w$ for $i$ odd.}

\end{algorithm2e}

Testing for equality can either be implemented by explicit equality tests at
each node of the search tree (which saves time when most elements end up in a
few large equality buckets) or by going down the search tree all the way to a
bucket $b_i$ ($i$ even) doing only $\leq$-comparisons, followed by a single
equality test with $x_{\frac{i}{2}}$, unless $i = 2v$. This last variant is
shown in Algorithm~\ref{alg:s5}, and the equality test is done in
line~\ref{alg:s5:eqtest}.

Postponing the equality test allows us to completely unroll the loop descending
the search tree (line~\ref{alg:s5:treeloop}), since there is no exit
condition. We can then also unroll the loop over the elements
(line~\ref{alg:s5:strloop}), interleaving independent tree descent
operations. The number of interleaved descents is limited in practice by the
number of registers to hold local variables like $i$ and $c$. As in
\cite{sanders2004super}, this is an important optimization since it allows the
instruction scheduler in a super scalar processor to parallelize the operations
by drawing data dependencies apart.

After reordering, the strings in the ``$< x_0$'' and ``$> x_{v-1}$'' buckets
$b_0$ and $b_{2v}$ keep common prefix length $h$. For other even buckets
$b_i$ the common prefix length is increased by $\lcp_x(\frac{i}{2})$.

An analysis similar to the one of multikey quicksort \cite{bentley1997fast} lets
us conjecture the following asymptotic time bound.

\begin{conjecture}\label{thm:s5}
  String sample sort with implicit binary trees, word parallelism and equality
  checking at each splitter node can be implemented to run in expected time
  $\Oh{\frac{D}{w} + n \log n}$.
\end{conjecture}

We now argue the correctness of Conjecture~\ref{thm:s5}, without giving a formal
proof\footnote{We are currently working on this for a final version of this
  paper}.  The classification schemes of multikey quicksort and string sample
sort can both be seen as a tree. In this tree, edges are either associated with
characters of a distinguishing prefix, or with string ranges determined by
splitters.

In multikey quicksort each inner node $z$ of the tree has three children: $<$,
$=$, and $>$. We can associate each character comparison during partitioning at
node $z$ with the thereby determined edge.  By selecting pivots randomly or
using a sample median, the expected number of $<$ and $>$~edges in all paths
from the root is $\Oh{\log n}$ \cite{hoare1962quicksort,bentley1997fast}, since
this approach is identical to atomic quicksort. Thus the time spent over all
comparisons accounted by these edges is expected $\Oh{n \log n}$. All
comparisons associated with $=$~edges correspond to characters from the
distinguishing prefix, and are thus bounded by $D$. In total we have $\Oh{D + n
  \log n}$ work in the multikey quicksort tree.

We can view string sample sort as a multikey quicksort using multiple pivots in
the classification tree, as seen in Figure~\ref{fig:ternary-tree}.  In string
sample sort, an $=$~edge matches $w$ characters, of which at least one is from
the distinguishing prefix $D$ (but usually all are). If any of the $w$
characters is not counted in $D$, then the $=$~edge leads to a leaf, which does
not require further sorting. There are at most $n$ such comparisons leading to
leaves, all other $=$~edges match $w$ characters. Thus we have at most
$\frac{D}{w} + n$ comparisons leading to $=$~edges.  To prove our conjecture, we
need to show that the expected number of $<$ and $>$~edges on all paths from the
root is $\Oh{\log n}$.  However, we are not aware of any analysis of sample sort
showing this expected run time for a \emph{fixed sample size}. Furthermore, in
string sample sort we have to deal with the probability of multiple equal
samples and need to resample strings repeatedly at higher depths, thus the known
analysis of a single top-level sampling approach
\cite{blelloch1991comparison,yang1987optimal,frazer1970samplesort} do not
apply. Nevertheless, due to repeated resampling, we can conjecture that the
bucket sizes grow small very fast, just as they do in atomic sample sort.  By
using the additional LCP information gained at $<$ and $>$ edges from
Equation~\ref{eq:lcp} one could decrease the expected path length from the root
further, though probably not asymptotically.

If the $=$~edges are taken immediately, as done in the variant with explicit
equality checking at each node, then we conjecture expected $\Oh{\frac{D}{w} + n
  \log n}$ time.  However, if we choose to unroll descents of the tree, then the
splitter at the root may match and the $\Th{\log v}$ additional steps down the
tree are superfluous.  This happens when many strings are identical, and the
corresponding splitters are high up in the tree. We thus have to attribute
$\Oh{(\frac{D}{w} + n) \log v}$ time to the $=$~edges. Together with the
expected cost of $<$ and $>$~edges, we conjecture in total an expected
$\Oh{(\frac{D}{w} + n) \log v + n \log n}$ bound.

String sample sort is particularly easy to parallelize for $p < n$, as in
current multi-core architectures, and we can state the following theorem.

\begin{theorem}
  A single step of super scalar string sample sort (Algorithm~\ref{alg:s5}) can
  be implemented to run on a CREW PRAM with $p$~processors in $\Oh{\frac{n}{p}
    \log v + \log pv + v}$ time and $\Oh{n \log v + pv}$ work.
\end{theorem}
\begin{proof}
  Sorting the sample requires $\Oh{\frac{a \log a}{p} + \log p}$ time and
  $\Oh{a \log a + p}$ work, where $a := \alpha v + \alpha - 1 \ll n$ is the
  sample size. Selecting the sample, picking splitters, constructing the tree
  and saving LCP of splitters is all $\Oh{\frac{a}{p}}$ time and $\Oh{a}$
  work. Each processors gets $\frac{n}{p}$ strings and in worst case runs down
  all $\log v$ steps in the classification tree, which is $\Oh{\frac{n}{p} \log
    v}$ time and $\Oh{n \log v}$ work. Departing from
  lines~\ref{alg:s5:bktzero}--\ref{alg:s5:redistribute}, each processor keeps
  its own bucket array $b_i$, initializes it in $\Oh{v}$ time, and classifies
  only those strings in its string set. Then, an interleaved global prefix sum
  over the $p(2v+1)$ bucket counters yields the boundaries in which each
  processor can independently redistribute its strings. The prefix sum runs in
  $\Oh{\log pv}$ time and $\Oh{pv}$ work~\cite{kogge1973parallel}, while
  counting and redistribution runs in $\Oh{\frac{n}{p}}$ time and $\Oh{n}$
  work. Summing all time and work yields our result.
\end{proof}

We only consider a single step here, and thus cannot use the distinguishing
prefix $D$ to bound the overall work.

\subsection{Implementation Details}

One goal of S$^5$ is to have a common classification data structure that fits
into the cache of all cores. Using this data structure, all PEs can
independently classify a subset of the strings into buckets in parallel.  The
process follows the classic distribution-based sorting steps: we first classify
strings (lines~\ref{alg:s5:strloop}--\ref{alg:s5:end-strloop}), counting how
many fall into each bucket (line~\ref{alg:s5:count}), then calculate a prefix
sum (line~\ref{alg:s5:prefixsum}) and redistribute the string pointers
accordingly (line~\ref{alg:s5:redistribute}). To avoid traversing the tree
twice, the bucket index of each string is stored in an oracle
(lines~\ref{alg:s5:end-strloop}, \ref{alg:s5:count},
\ref{alg:s5:redistribute}). Additionally, to make higher use of super scalar
parallelism, we even separate the classification loop
(line~\ref{alg:s5:strloop}) from the counting loop (line~\ref{alg:s5:count}), as
done by \cite{karkkainen2009engineering}.

Like in S$^4$, the binary tree of splitters is stored in level-order as an array
$t$ (line~\ref{alg:s5:tree}), allowing efficient traversal using $i := 2 i +
\set{0,1}$, without branch mispredictions in line~\ref{alg:s5:chartest}. The
pseudo-code ``$(c \leq t_i)$'', which yields $0$ or $1$, can be implemented
using different machine instructions. One method is to use the instruction
\texttt{SETA}, which sets a register to $0$ or $1$ depending on a preceding
comparison. Alternatively, newer processors have predicated instruction like
\texttt{CMOVA} to conditionally move one register to another, again depending on
a preceding comparison's outcome. We noticed that \texttt{CMOVA} was slightly
faster than flag arithmetic.

While traversing the classification tree, we compare $w$ characters using one
arithmetic comparison. However, we need to make sure that these comparisons
have the desired outcome, e.g., that the most significant bits of the register
hold the first character. For little-endian machines and 8-bit characters, which
are used in all of our experiments, we need to \emph{swap the byte order} when
loading character from a string. In our implementation we do this using the
\texttt{BSWAP} machine instruction. In the pseudo-code (Algorithm~\ref{alg:s5})
this operation is symbolized by $\FuncSty{chars}_h(s_i)$, which fetches $w$
characters from $s_i$ at depth $h+1$, and swaps them appropriately.

For performing the equality check, already mentioned in the previous section, we
want to discuss four different alternatives in more technical details here:
\begin{enumerate}
\item One can traverse the tree using only $\leq$-comparisons and perform the
  equality check afterwards, as shown in Algorithm~\ref{alg:s5}. For this we
  keep the splitters $x_i$ in an in-order array, in addition to the
  classification tree $t$, which contains them in level-order. Duplicating the
  splitters avoids additional work in line~\ref{alg:s5:eqtest}, where $i$ is an
  in-order index. This variant was our final choice, called S$^5$-Unroll, as it
  was overall fastest.

\item The additional in-order array from the previous variant, however, can be
  removed. Instead, a rather simple calculation involving only bit operations
  can be used to transform the in-order index $i$ back to level-order, and reuse
  the classification tree $t$. We tried this variant, but found no performance
  advantage over the first.

\item Another idea is to keep track of the last $\leq$-branch during tree
  traversal, this however was slower and requires an extra register for each of
  the interleaved descents.

\item The last variant is to check for equality after each comparison in
  line~\ref{alg:s5:chartest}. This requires only an additional \texttt{JE}
  instruction and no extra \texttt{CMP} in the inner-most loop.  The branch
  misprediction cost of the \texttt{JE} is counter-balanced by skipping the rest
  of the tree.  As $i$ is a tree-order index when exiting the inner loop, we
  need to apply the inverse of the transformation mentioned in the second method
  to $i$ to determine the correct equality bucket.  Thus in this fourth variant,
  named S$^5$-Equal, no additional in-order splitter array is needed.

\end{enumerate}

The sample is drawn pseudo-randomly with an oversampling factor $\alpha = 2$ to
keep it in cache when sorting with STL's introsort and building the search tree.
Instead of using the straight-forward equidistant method to draw splitters from
the sample, as shown in Algorithm~\ref{alg:s5} (line~\ref{alg:s5:sample}), we
developed a simple recursive scheme that tries to avoid using the same splitter
multiple times: Select the middle sample $m$ of a range $a..b$ (initially the
whole sample) as the middle splitter $\bar{x}$. Find new boundaries $b'$ and
$a'$ by scanning left and right from $m$ \emph{skipping} samples equal to
$\bar{x}$. Recurse on $a..b'$ and $a'..b$.  The splitter tree selected by this
heuristic was never slower than equidistant selection, but slightly faster for
inputs with many equal common prefixes. It is used in all our experiments.

The LCP of two consecutive splitters in line~\ref{alg:s5:slcp} can be calculated
without a loop using just two machine instructions: \texttt{XOR} and
\texttt{BSR} (to count the number of leading zero bits in the result of
\texttt{XOR}). In our implementation, these calculation are done while selecting
splitters. Similarly, we need to check if splitters contain end-of-string
terminators, and skip the recursion in this case.

For current 64-bit machines with 256\,KiB L2 cache, we use $v = 8191$. Note that
the limiting data structure which must fit into L2 cache is not the splitter
tree $t$, which is only 64\,KiB for this $v$, but is the bucket counter array
$b$ containing $2v+1$ counters, each 8 bytes long. We did not look into methods
to reduce this array's size, because the search tree is stored both in
level-order and in in-order, and thus we could not increase the tree size
anyway.

\subsection{Practical Parallelization of \texorpdfstring{S$^5$}{S5}}\label{sec:parallel-s5}

Parallel S$^5$ (pS$^5$) is composed of four sub-algorithms for differently sized
subsets of strings. For a string subset $\Strings$ with $|\Strings| \geq
\frac{n}{p}$, a \emph{fully parallel version} of S$^5$ is run, for large sizes
$\frac{n}{p} > |\Strings| \geq t_m$ a sequential version of S$^5$ is used, for
sizes $t_m > |\Strings| \geq t_i$ the fastest sequential algorithm for
medium-size inputs (caching multikey quicksort from Section~\ref{sec:para-mkqs})
is called, which internally uses insertion sort when $|\Strings| < t_i$.  We
empirically determined $t_m = 1\,\text{Mi}$ and $t_i = 64$ as good thresholds to
switch sub-algorithms.

The fully parallel version of S$^5$ uses $p' = \Theta( \frac{|\Strings|}{p} )$
threads for a subset $\Strings$. It consists of four stages: selecting samples
and generating a splitter tree, parallel classification and counting, global
prefix sum, and redistribution into buckets. Selecting the sample and
constructing the search tree are done sequentially, as these steps have
negligible run time. Classification is done independently, dividing the string
set evenly among the $p'$ threads. The prefix sum is done sequentially once all
threads finish counting.

In both the sequential and parallel versions of S$^5$ we permute the string
pointer array using out-of-place redistribution into an extra array. In
principle, we could do an in-place permutation in the sequential version by
walking cycles of the permutation \cite{mcilroy1993engineering}. Compared to
out-of-place copying, the in-place algorithm uses fewer input/output streams and
requires no extra space. However, we found that modern processors optimize the
sequential reading and writing pattern of the out-of-place version better than
the random access pattern of the in-place walking. Furthermore, for fully
parallel S$^5$, an in-place permutation cannot be done in the same manner.  We
therefore always use \emph{out-of-place redistribution}, with an extra string
pointer array of size $n$. For recursive calls, the role of the extra array and
original array are swapped, which saves superfluous copying work.

All work in parallel S$^5$ is dynamically load balanced via a central job
queue. We use the lock-free queue implementation from Intel's Thread Building
Blocks (TBB) and threads initiated by OpenMP to create a \emph{light-weight
  thread pool}.

To make work balancing most efficient, we modified all sequential
sub-al\-go\-rithms of parallel S$^5$ to use an explicit recursion stack. The
traditional way to implement dynamic load balancing would be to use work
stealing among the sequentially working threads. This would require the
operations on the local recursion stacks to be synchronized or atomic. However,
for our application fast stack operations are crucial for performance as they
are very frequent. We therefore choose a different method: \emph{voluntary work
  sharing}. If the global job queue is empty and a thread is idle, then a global
atomic counter is incremented to indicate that other threads should share their
work. These then free the stack level with the \emph{largest subproblems} from
their local recursion stack and enqueue these as separate, independent
jobs. This method avoids costly atomic operations on the local stacks, replacing
it by a faster counter check, which itself \emph{need not be synchronized} or
atomic. The short wait of an idle thread for new work does not occur often,
because the largest recursive subproblems are shared. Furthermore, the global
job queue never gets large because most subproblems are kept on local stacks.

\section{Parallel Multiway LCP-Mergesort}\label{sec:para-mergesort}

When designing pS$^5$ we considered L2 cache sizes, word parallelism, super
scalar parallelism and other modern features. However, new architectures with
large amounts of RAM are now commonly non-uniform memory access (NUMA) systems,
and the RAM chips are distributed onto different memory banks, called \emph{NUMA
  nodes}. In preliminary synthetic experiments, access to memory on ``remote''
nodes was 2--5 times slower than memory on the local socket, because the
requests must pass over an additional interconnection bus. This latency and
throughput disparity brings algorithms for external and distributed memory to
mind, but the divide is much less pronounced and block sizes are smaller.

In light of this disparity, we propose to use \emph{independent string sorters}
on each NUMA node, and then \emph{merge} the sorted results. During merging, the
amount of information per transmission unit passed via the interconnect (64-byte
cache lines) should be maximized. Thus, besides the sorted string pointers, we
also want to use LCP information to skip over known common prefixes, and cache
the distinguishing characters.

While merging sorted sequences of strings with associated LCP information is a
very intuitive idea, remarkably, only one very recent paper by Ng and
Kakehi~\cite{ng2008merging} fully considers LCP-aware mergesort for
strings. They describe \emph{binary} LCP-merge\-sort and perform an average case
analysis yielding estimates for the number of comparisons needed. For the NUMA
scenario, however, we need a \emph{parallel $K$-way LCP-merge}, where $K$ is
the number of NUMA nodes.  Furthermore, we also need to extend our existing
string sorting algorithms to save the LCP array.

In the next section, we first review binary LCP-aware merging. On this
foundation we then propose and analyze parallel $K$-way LCP-merging in
Sections~\ref{sec:merge-kway}--\ref{sec:merge-details}. For node-local LCP
calculations, we extended pS$^5$ appropriately, and describe the necessary
LCP-aware base case sorter in Section~\ref{sec:lcp-inssort}. Further information
on the results of this section are available in the bachelor thesis of Andreas
Eberle~\cite{eberle2014parallel}.

\subsection{Binary LCP-Compare and LCP-Mergesort}\label{sec:mergesort-binary}

We reformulate the binary LCP-merge and -mergesort presented by Ng and
Ka\-ke\-hi~\cite{ng2008merging} here in a different way. Our exposition is somewhat
more verbose than necessary, but this is intentional and prepares for a simpler
description of $K$-way LCP-merge in the following section.

Consider the basic comparison of two strings $s_a$ and $s_b$. If there is no
additional LCP information, the strings must be compared character-wise until a
mismatch is found. However, if we have additionally the LCP of $s_a$ and $s_b$
to another string $p$, namely $\lcp(p,s_a)$ and $\lcp(p,s_b)$, then we can first
compare these LCP values. Since both reference $p$, we know that $s_a$ and $s_b$
share a common prefix $\min\{ \lcp(p,s_a), \lcp(p,s_b) \}$ and that this common
prefix is maximal (i.e. longest). Thus if $\lcp(p,s_a) < \lcp(p,s_b)$, then the
two strings $s_a$ and $s_b$ differ at position $\ell := \lcp(p,s_a) + 1$. If we
now furthermore assume $p \leq s_a$, then we immediately see $p[\ell] =
s_b[\ell] < s_a[\ell]$, from which follows $s_b < s_a$. The argument can be
applied symmetrically if $\lcp(p,s_b) < \lcp(p,s_a)$.

There remains the case $\lcp(p,s_a) = \lcp(p,s_b)$. Here, the LCP information
only reveals that both have a common prefix $\lcp(p,s_a)$, and additional
character comparisons starting at the common prefix are necessary to order the
strings.

\begin{algorithm2e}[t]
\caption{Binary LCP-Compare}\label{alg:LCP-compare}\normalsize
\Function{\LCPCompare{$(a,s_a,h_a), (b,s_b,h_b)$}}
{
  \KwIn{$(a,s_a,h_a)$ and $(b,s_b,h_b)$ where $s_a$ and $s_b$ are two strings
    together with LCPs $h_a = \lcp(p,s_a)$ and $h_b = \lcp(p,s_b)$, and $p$ is
    another string with $p \leq s_a$ and $p \leq s_b$.}

  \uIf(\Remi{case 1: LCPs are equal $\Rightarrow$ compare more characters,})
  {$h_a = h_b$}
  {
    $h' := h_a$ \Rem{starting at $h' = h_a = h_b$.}
    \While(\Remi{Compare characters and}\nllabel{alg:LCP-compare:charloop})
    {$(s_a[h'] \neq 0 \And s_a[h'] = s_b[h'])$}
    {
      $h' \Inc$ \Rem{increase total LCP.}
    }
    \lIf(\nllabel{alg:LCP-compare:charcmp2})
    {$s_a[h'] \leq s_b[h']$}
    {
      \Return $(a,h_a,b,h')$
    }
    \lElse
    {
      \Return $(b,h_b,a,h')$
    }
  }
  \lElseIf(\Remi{case 2: $s_b[h_a {+} 1] < s_a[h_a {+} 1]$.})
  {$h_a < h_b$}
  {
    \Return $(b,h_b,a,h_a)$
  }
  \lElse(\Remi{case 3: $s_a[h_b {+} 1] < s_b[h_b {+} 1]$.})
  {
    \Return $(a,h_a,b,h_b)$
  }
  \KwOut{$(x,h_x,y,h')$ with $s_x \leq s_y$, $\{x,y\} = \{a,b\}$, and $h' = \lcp(s_a,s_b)$.}
}
\end{algorithm2e}

The pseudo-code in Algorithm~\ref{alg:LCP-compare} implements these three
cases. In preparation for $K$-way LCP-merge, the function \LCPCompare
additionally takes variables $a$ and $b$, which are corresponding indexes and
returns these instead of $s_a$ or $s_b$. It also calculates more information
than just the order of $s_a$ and $s_b$, since future LCP-aware comparisons also
require $\lcp(s_a,s_b)$.

In the cases $\lcp(p,s_a) \neq \lcp(p,s_b)$, the $\lcp(s_a,s_b)$ is easily
inferred since the character after the smaller LCP differs in $s_a$ and
$s_b$. From this follows $\lcp(s_a,s_b) = \min\{ \lcp(p,s_a), \lcp(p,s_b) \}$, as already
stated above. For $\lcp(p,s_a) = \lcp(p,s_b)$ each additionally compared equal
character is common to both $s_a$ and $s_b$, and the comparison loop in
line~\ref{alg:LCP-compare:charloop} of Algorithm~\ref{alg:LCP-compare} breaks at
the first mismatch or zero termination. Thus afterwards $h' = \lcp(s_a,s_b)$,
and can be returned as such.

Using \LCPCompare we can now build a binary LCP-aware merging method, which
merges two sorted string sequences with associated LCP arrays. One only needs to
take $s_a$ and $s_b$, compare then using \LCPCompare, write the smaller of them,
say $s_a$, to the output and fetch its successor $s'_a$ from the sorted
sequence. The written string $s_a$ then plays the role of $p$ in the
discussion above, and the next two candidate strings $s'_a$ and $s_b$ can be
compared, since $\lcp(p,s_b) = \lcp(s_a,s_b)$ is returned by \LCPCompare and
$\lcp(p,s'_a) = \lcp(s_a,s'_a)$ is known from the corresponding LCP array. This
procedure is detailed in Algorithm~\ref{alg:LCP-merge}. For binary merging, we
can ignore the $h_x$ returned by \LCPCompare. Notice that using the indexes $x$
and $y$, the LCP invariant can be restored using just one assignment in
line~\ref{alg:LCP-merge:lcp-swap}.

\begin{theorem}\label{thm:LCP-mergesort}
  Using Algorithm~\ref{alg:LCP-merge}, one can implement a binary LCP-mergesort
  algorithm, which requires at most $L + n \lceil \log_2 n \rceil$ character
  comparisons and runs in $\Oh{D + n \log n}$ time.
\end{theorem}
\begin{proof}
  We assume the divide step of binary LCP-mergesort to do straight-forward
  halving as in non-LCP mergesort \cite{knuth1998sorting}, which is why we
  omitted its pseudo-code. Likewise, the recursive division steps have at most
  depth $\lceil \log_2 n \rceil$ when reaching the base case. If we briefly
  ignore the character comparison loop in \LCPCompare,
  line~\ref{alg:LCP-compare:charloop}, and regard it as a single comparison,
  then the standard divide-and-conquer recurrence $T(n) \leq T(\lfloor
  \frac{n}{2} \rfloor) + T(\lceil \frac{n}{2} \rceil) + n$ of non-LCP mergesort
  still holds. Regarding the character comparison loop, we can establish that
  each increment of $h'$ ultimately increases the overall LCP sum by exactly
  one, since in all other statements LCPs are only moved, swapped or stored, but
  never decreased or discarded.  Another way to see this is that the character
  comparison loop is the only place where characters are compared, thus to be
  able to establish the correctly sorted order, all distinguishing characters
  must be compared here.

  We regard the three different comparison expressions in
  lines~\ref{alg:LCP-compare:charloop}--\ref{alg:LCP-compare:charcmp2} as one
  ternary comparison, as the same values are checked again and zero-terminators
  can be handled using flag tests.  To count the total number of comparisons, we
  can thus account for all \textsl{true}-outcomes of the while loop condition in
  \LCPCompare (line~\ref{alg:LCP-compare:charloop}) using $L$, and all
  \textsl{false}-outcomes using $n \lceil \log_2 n \rceil$, since this is the
  highest number of times case 1 can occur in the mergesort recursion. This is
  an upper bound, and for most string sets, cases~2 and 3 reduce the number of
  comparisons in the second term.  Since $L \leq D$, the time complexity $\Oh{D
    + n \log n}$ follows immediately.
\end{proof}

\begin{algorithm2e}[t]
\caption{Binary LCP-Merge}\label{alg:LCP-merge}\normalsize

\KwIn{$\Strings_1$ and $\Strings_2$ two sorted sequences of strings with LCP
  arrays $H_1$ and $H_2$. Assume sentinels $\Strings_k[|\Strings_k|+1] = \infty$
  for $k = 1,2$, and $\Strings_0[0] = \varepsilon$.}

$i_1 := 1$,\quad $i_2 := 1$,\quad $j := 1$ \Rem{Indexes for $\Strings_1$, $\Strings_2$ and $\Strings_0$.}
$h_1 := 0$,\quad $h_2 := 0$ \Rem{Invariant: $h_k = \lcp(\Strings_k[i_k], \Strings_0[j-1])$ for $k=1,2$.}
\While{$i_1 + i_2 \leq |\Strings_1| + |\Strings_2|$}
{
  $s_1 := \Strings_1[i_1]$,\quad $s_2 := \Strings_2[i_2]$ \Rem{Fetch strings $s_1$ and $s_2$,}
  $(x,\bot,y,h') = \LCPCompare((1,s_1,h_1), (2,s_2,h_2))$ \Rem{compare them,}
  $(\Strings_0[j],H_0[j]) := (s_x,h_x)$,\quad $j \Inc$ \Rem{put smaller into output}
  $i_x \Inc$,\quad $(h_x,h_y) := (H_x[i_x],h')$ \Rem{and advance to next.}\nllabel{alg:LCP-merge:lcp-swap}
}

\KwOut{$\Strings_0$ contains sorted $\Strings_1$ and $\Strings_2$, and $\Strings_0$ has the LCP array $H_0$}

\end{algorithm2e}

Ng and Kakehi \cite{ng2008merging} do not give an explicit worst case
analysis. Their average case analysis shows, that the total number of character
comparisons of binary LCP-mergesort is about $n (\mu_a - 1) + P_\omega n \log_2
n$, where $\mu_a$ is the average length of distinguishing prefixes and
$P_\omega$ the probability of a ``breakdown'' (which corresponds to case 1 in
\LCPCompare). Taking $P_\omega = 1$ and $\mu_a = \frac{D}{n}$, their equation
matches our worst-case result, except for the minor difference between $D$ and
$L$.

\subsection{\texorpdfstring{$K$}{K}-way LCP-Merge}\label{sec:merge-kway}

To accelerate LCP-merge for NUMA systems, we extended the binary LCP-merge
approach to $K$-way LCP-merge using tournament
trees~\cite{knuth1998sorting,sanders00fast}, since current NUMA systems have
four or even eight nodes. We could not find any reference to $K$-way LCP-merge
in the literature, even though the idea to store and reuse LCP information
inside the tournament tree is very intuitive. The algorithmic details, however,
require precise elaboration.

\begin{figure}\centering\normalsize
  \begin{tikzpicture}[
  yscale=-1,
  xscale=0.8,yscale=0.8,
  keynode/.style={anchor=base,inner sep=1pt},
  uparrow/.style={->},
  txtnode/.style={anchor=base east},
  ]

  \node (n0) [keynode] at (0,-1) {$(w,h_1)$};

  \node[txtnode] at (-1,-1) {Winner};

  \node (n1) [keynode] at (0,0) {$(y_2,h_2)$};

  \node (n2) [keynode] at (-2,1) {$(y_3,h_3)$};
  \node (n3) [keynode] at (2,1) {$(y_4,h_4)$};

  \node[txtnode,anchor=east] at (-2.7,0.25) {Losers};

  \draw[densely dotted] (-5.8,1.4) -- (4,1.4);

  \node[txtnode] at (-4,2) {Players};

  \node (n4) [keynode] at (-3,2) {$(s_1,h'_1)$};
  \node (n5) [keynode] at (-1,2) {$(s_2,h'_2)$};
  \node (n6) [keynode] at (+1,2) {$(s_3,h'_3)$};
  \node (n7) [keynode] at (+3,2) {$(s_4,h'_4)$};

  \draw[uparrow] (n1) -- (n0);
  \draw[uparrow] (n2) -- (n1);
  \draw[uparrow] (n3) -- (n1);
  \draw[uparrow] (n4) -- (n2);
  \draw[uparrow] (n5) -- (n2);
  \draw[uparrow] (n6) -- (n3);
  \draw[uparrow] (n7) -- (n3);

  \node[txtnode] at (-4,3) {Inputs};

  \node (s1) [keynode] at (-3,3) {$(\Strings_1,H_1)$};
  \node (s2) [keynode] at (-1,3) {$(\Strings_2,H_2)$};
  \node (s3) [keynode] at (+1,3) {$(\Strings_3,H_3)$};
  \node (s4) [keynode] at (+3,3) {$(\Strings_4,H_4)$};

  \draw[uparrow] (s1) -- (n4);
  \draw[uparrow] (s2) -- (n5);
  \draw[uparrow] (s3) -- (n6);
  \draw[uparrow] (s4) -- (n7);

  \node[txtnode] at (-1,-2) {Output};

  \node (s0) [keynode] at (0,-2) {$(\Strings_0,H_0)$};

  \draw[uparrow] (n0) -- (s0);

  \draw[densely dotted] (-5.8,-1.6) -- (4,-1.6);

\end{tikzpicture}
  \caption{LCP-aware tournament tree with $K=4$ showing input and output
    streams, their front items as players, the winner node $(w,h_1)$, and loser
    nodes $(y_i,h_i)$, where $y_i$ is the index of the losing player of the
    particular game and $h_i$ is the LCP of $s_{y_i}$ and the winner of the
    comparison at node $i$.}\label{fig:lcp-losertree}
\end{figure}
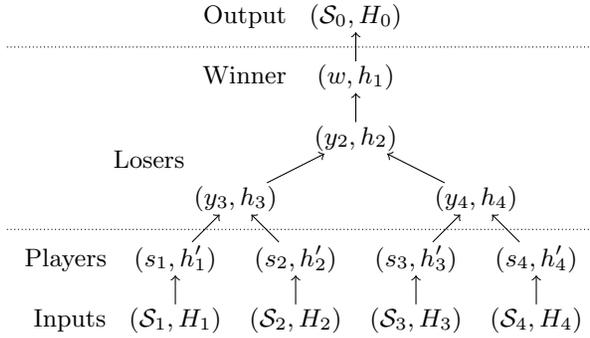

As commonly done in multiway mergesort, to perform $K$-way merging one regards
selection of the next item as a tournament with $K$ players (see
Figure~\ref{fig:lcp-losertree}). Players compete against each other using binary
comparisons, and these games are organized in a binary tree. Each node in the
tree corresponds to one game, and we label the nodes of the tree with the
``losers'' of that particular game. The ``winner'' continues upward and plays
further games, until the overall winner is determined. The winner is commonly
placed on the top, in an additional node, and with this node, the tournament
tree contains each player exactly once.  Hence the tree has exactly $K$ nodes,
since we do not consider the input, output or players part of the tree.  For
sorting strings into ascending sequences, the ``overall winner'' of the
tournament is the lexicographically smallest string.

The first winner is determined by playing an initial round on all $K$ nodes from
bottom up.  This winner can then be sent to the output, and the next item from
the corresponding input sequence takes its place.  Thereafter, only $\log_2 K$
games must be replayed, since the previous winner only took part in those games
along the path from the corresponding input to the root of the tournament
tree. This can be repeated until all streams are empty. By using sentinels for
empty inputs, special cases can be avoided, and we can assume $K$ to be a power
of two, filling up with empty inputs as needed. Thus the tournament tree can be
assumed to be a perfect binary tree, and can be stored implicitly in an
array. Navigating upward in the tree corresponds to division by two: $\lceil
\frac{i}{2} \rceil$ is the parent of child $i$, unless $i=1$ (note that we use a
one-based array here). Thus finding the path from input leaf to root when
replaying the game can be implemented very efficiently. Inside the tree nodes,
we save the loser \emph{input index} $y_i$, or winner index $w$ (renamed from
$y_1$), instead of storing the string $s_i$ or a reference thereof.

We now discuss how to make the tournament tree LCP-aware. The binary comparisons
between players are done using \LCPCompare (Algorithm~\ref{alg:LCP-compare}),
which may perform explicit character comparisons in case 1. Since we want to
avoid comparing characters already found equal, we store alongside the loser
input index $y_i$ an LCP value $h_i$ in the tree node. The LCP $h_i$ represents
the LCP of the stored losing string $s_{y_i}$ with the particular game's winner
string, which passes upward to play further comparisons. If we call the
corresponding winner $x_i$, even though it's not explicitly stored, then $h_i =
\lcp(s_{x_i},s_{y_i})$.

After the initial overall winner $w$ is determined, we have to check that all
requirements of \LCPCompare are fulfilled when replaying the games on the path
from input $w$ to the root.  The key argument is that the overall winner $w$ was
also the winner of all individual games on the path.  Hence, for all games $i$
on that path $h_i = \lcp(s_w,s_{y_i})$.  Thus after writing $s_w$ to the output,
and advancing to the next item $(s'_w,h''_w)$ from the input $(\Strings_w,H_w)$,
we have $p = s_w$ as the common, smaller predecessor string.  The previous
discussion about the overall winner $w$ is also valid for the individual winner
$x_i$ of any node $i$ in the tree, since it is the winner of all games leading
from input $x_i$ to node $i$.

The function signature $(x,h_x,y,h_y) = \LCPCompare((a,s_a,h_a),(b,s_b,h_b))$
was designed to be played on two nodes $(a,h_a)$ and $(b,h_b)$ of the LCP-aware
tournament tree. When replaying a path, we can picture a node $(a,h_a)$ moving
``upward'' along the edges. \LCPCompare is called with this moving node and the
loser information $(b,h_b) := (y_i,h_i)$ saved in the encountered node $i$.
After performing the comparisons, the returning values $(x,h_x)$ are the winner
node, which passes upwards, and $(y,h_y)$ are the loser information, which is
saved in the node $i$.  Thus \LCPCompare effectively selects the winner of each
game, and computes the loser information for future LCP-aware comparisons.  Due
to the recursive property discussed in the previous paragraph, the requirements
of \LCPCompare remains valid along all paths, and \LCPCompare can switch between
them.

\begin{algorithm2e}[t]
\caption{$K$-way LCP-Merge}\label{alg:kway-LCP-merge}\normalsize

\KwIn{$\Strings_1,\ldots,\Strings_K$ sorted sequences of strings with LCP arrays
  $H_1,\ldots,H_K$ and common prefix $\overline{h}$. Assume sentinels
  $\Strings_k[|\Strings_k|+1] = \infty$ for $k = 1,\ldots,K$, and $K$ a power of
  two.}

$i_k := 1 \;\forall\, k = 1,\ldots,K$,\quad $j := 1$ \Rem{Initialize indexes for $\Strings_1,\ldots,\Strings_K$ and $\Strings_0$.}
\For(\Remi{Initialize loser tree, building})
{$k := 1,\ldots,K$}
{
  $s_k := \Strings_k[i_k]$ \Rem{perfect subtrees left-to-right.}
  $(x,h') := (k,\overline{h})$,\quad $v := K + k$ \Rem{Play from input node $v$, upward till the root}
  \While(\Remi{of a perfect odd-based subtree is reached.})
  {$v \text{ is even}$}
  {
    $v := \frac{v}{2}$, \quad$(x,h',y_v,h_v) := \LCPCompare((x,s_x,h'), (y_v,s_{y_v},h_v))$ \;
  }
  $(y_v,h_v) := (x,h')$ \Rem{Save intermediate winner in odd node at top.}
}
$w := y_1$ \Rem{Initial winner after all games (rename $y_1 \rightarrow w$).}
\While(\Remi{Loop until output is done.})
{$j \leq \sum_{k=1}^K |\Strings_k|$}
{
  $(\Strings_0[j],H_0[j]) := (s_w,h_1)$,\quad $j \Inc $ \Rem{Output winner string $s_w$ with LCP $h_1$.}
  $i_w \Inc$,\quad $s_w := \Strings_w[i_w]$ \Rem{Replace winner with next item from input.}
  $(x,h') := (w,H_w[i_w])$,\quad $v := K + w$ \Rem{Play from input node $v$, all games}
  \While(\Remi{upward to root (unrollable loop).})
  {$v > 2$}
  {
    $v := \lceil \frac{v}{2} \rceil$,\quad $(x,h',y_v,h_v) := \LCPCompare((x,s_x,h'), (y_v,s_{y_v},h_v))$ \;
  }
  $(w,h_1) := (x,h')$ \Rem{Save next winner at top.}
}
\KwOut{$\Strings_0$ contains sorted $\Strings_1,\ldots,\Strings_K$ and has the LCP array $H_0$}

\end{algorithm2e}

This LCP-aware $K$-way merging procedure is shown in pseudo-code in
Algorithm~\ref{alg:kway-LCP-merge}.  We build the initial tournament tree
incrementally from left to right, playing all games only on the right-most path
of every odd-based perfect subtree.  This right-most side contains only nodes
with even index.

The following theorem considers only a single execution of $K$-way LCP-merg\-ing,
since this is what we need in our NUMA scenario:

\begin{theorem}\label{thm:kway-LCP-merge}
  Algorithm~\ref{alg:kway-LCP-merge} requires at most $\Delta L + n \log_2 K +  
  K$ character comparisons, where $n = |\Strings_0|$ is the total number of
  strings and $\Delta L = L(H_0) - \sum_{k=1}^K L(H_k)$ is the sum of increments
  to LCP array entries.
\end{theorem}
\begin{proof}
  We focus on the character comparisons in the sub-function \LCPCompare, since
  Algorithm~\ref{alg:kway-LCP-merge} itself does not contain any character
  comparisons. As in the proof of Theorem~\ref{thm:LCP-mergesort}, we can
  account for all \textsl{true}-outcomes of the while loop condition in
  \LCPCompare (line~\ref{alg:LCP-compare:charloop}) using $\Delta L$, since it
  increments the overall LCP. We can bound the number of \textsl{false}-outcomes
  by bounding the number of calls to \LCPCompare, which occurs exactly $K$ times
  when building the tournament tree, and then $\log_2 K$ times for each of the
  $n$ output string (we actually have one superfluous run in the pseudo-code,
  but we keep it as it makes the code shorter).  As before, this upper bound,
  $\Delta L + n \log_2 K + K$, is only attained in pathological cases, and for
  most inputs, cases 2 and 3 in \LCPCompare reduce the overall number of
  character comparisons.
\end{proof}

\begin{theorem}
  Using Algorithm~\ref{alg:kway-LCP-merge} one can implement a $K$-way
  LCP-mergesort algorithm, which requires less than $L + n \lceil \log_K n
  \rceil \log_2 K + n + \lceil \frac{n-1}{K-1} \rceil$ character comparisons and runs in
  $\Oh{D + n \log n}$ time.
\end{theorem}
\begin{proof}
  We assume the divide step of $K$-way LCP-mergesort to split into $K$
  sub-problems of nearly even size. Using Theorem~\ref{thm:kway-LCP-merge}
  yields the recurrence $T(n) = K \cdot T(\frac{n}{K}) + n \log_2 K + K$ with
  $T(1) = 0$, if we ignore the character comparisons loop.  Assuming $n = K^d$
  for some integer $d$, the recurrence can be solved elementary using induction,
  yielding $T(n) = n \log_K n \cdot \log_2 K + \frac{K (n-1)}{K-1}$. For $n \neq
  K^d$, the input cannot be split evenly into recursive subproblems. However, to
  keep this analysis simple, we use $K$-way mergesort even when $n < K$, and
  thus incur the cost of Theorem~\ref{thm:kway-LCP-merge} also at the base
  level. So, we have $\lceil \log_K n \rceil$ levels of recursion. As in
  previous proofs, we account for all matching character comparisons with $L$,
  and all others with the highest number of occurrences of case 1 in \LCPCompare
  in the whole recursion, which is $T(n)$. Since $L \leq D$, the run time
  follows.
\end{proof}

In the proof we assume $K$-way LCP-merge even in the base level. In an
implementation, one would chose a different merger when $n < K$. By selecting
$2$-way LCP-merge, the number of comparisons in the lowest recursion is reduced,
and we can get a bound of $L + n \log_2 n + \oh{nK}$, which is close to the
one in Theorem~\ref{thm:kway-LCP-merge}.

\subsection{Practical Parallelization of \texorpdfstring{$K$}{K}-way LCP-Merge}\label{sec:merge-kway-parallel}

We now discuss how to parallelize $K$-way LCP-merge when given $K$ sorted input
streams. The problem is that merging itself cannot be parallelized without
significant overhead~\cite{cole1988parallel}, as opposed to the classification
and distribution in pS$^5$.  Instead, we want to split the problem into disjoint
areas of independent work, as done commonly in practical parallel
multiway mergesort sorting algorithms and implementations
\cite{akl1987optimal,singler2007mcstl}.

In contrast to atomic merging, a perfect split with respect to the number of
elements in the subproblems by no means guarantees good load balance for string
merging. Rather, the amount of work in each piece depends on the unknown values
of the common prefixes. Therefore, dynamic load balancing is needed anyway and
we can settle for a simple and fast routine for splitting the input into
pieces that are small enough to allow good load balance. We now outline our
current approach\footnote{which we intend to improve for the final version.}.
Since access to string characters
incurs costly cache faults, we want to use the information in the LCP array to
help split the input streams.  In principle, in the following heuristic we merge
the top of the LCP interval trees \cite{abouelhoda2004replacing} of the $K$
input streams to find independent areas.

If we consider all occurrences of the global minimum in an LCP array, then these
split the input stream into disjoint areas starting with the same distinct
prefix. The only remaining challenge is to match equal prefixes from the $K$
input streams, and for this matching we need to inspecting the first
distinguishing characters of any string in the area. Matching areas can then be
merged independently.

Depending on the input, considering only the global LCP minima may not yield
enough independent work. However, we can apply the same splitting method again
on matching sub-areas, within which all strings have a longer common prefix, and
the global minimum of the sub-area is larger.

We put these ideas together in a splitting heuristic, which scans the $K$ input
LCP arrays sequentially once, and creates merge jobs while scanning.  We start
by reading $w$ characters from the first string of all $K$ input streams, and
select those inputs with the smallest character block $\overline{c}$. In each of
these selected inputs, we scan the LCP array forward, skipping over all entries
$> w$, and checking entries $= w$ for equal character blocks, until either an
entry $< w$ or a mismatching character block is found. This forward scan
encompasses all strings with prefix $\overline{c}$, and an independent merge job
can be started.  The process is then repeated with the next strings on all $K$
inputs.

We start the heuristic with $w = 8$ (loading a 64-bit register full of
characters), but reduce $w$ depending on how many jobs are started, as otherwise
the heuristic may create too many splits, e.g. for random input strings. We
therefore calculate an expected number of independent jobs, and adapt $w$
depending on how much input is left and how many jobs were already created. This
adaptive procedure keeps $w$ high for inputs with high average common prefix and
low otherwise.

We use the same load balancing framework as with pS$^5$ (see
Section~\ref{sec:parallel-s5}).  During merge jobs, we check if other threads
are idle via the global unsynchronized counter variable. To reduce balancing
overhead, we check only every 4\,Ki processed strings. If idle threads are
detected, then a $K$-way merge job is split up into further independent jobs
using the same splitting heuristic, except that a common prefix of all strings
may be known, and is used to offset the character blocks of size $w$.

\subsection{Implementation Details}\label{sec:merge-details}

Our experimental platforms have $m \in \{ 4, 8 \}$ NUMA nodes, and we use
parallel $K$-way LCP-merge only as a top-level merger on $m$ input
streams. Thus we assume the $N$ inputs characters to be divided evenly onto the
$m$ memory nodes.  On the individual NUMA memory nodes, we pin about
$\frac{p}{m}$ threads and run pS$^5$ on the string subset.

Since $K$-way LCP-merge requires the LCP arrays of the sorted sequences, we
extended pS$^5$ to optionally save the LCP value while sorting. The string
pointers and LCP arrays are kept separate, as opposed to interleaving them as
``annotated'' strings~\cite{ng2008merging}. This was done, because pS$^5$
already requires an additional pointer array during out-of-place
redistribution. The additional string array and the original string array are
alternated between in recursive calls. When a subset is finally sorted, the
correctly ordered pointers are copied back to the original array, if
necessary. This allows us to place the LCP values in the additional array.

The additional work and space needed by pS$^5$ to save the LCP values is very
small, we basically get LCPs for free. Most LCPs are calculated in the base case
sorter of pS$^5$, and hence we describe LCP-aware insertion sort in the next
section. All other LCPs are located at the boundaries of buckets separated by
either multikey quicksort or string sample sort. We calculate these boundary
LCPs after recursive levels are finished, and use the saved splitters or pivot
elements whenever possible.

The splitting heuristic of parallel $K$-way LCP-merge creates jobs with varying
$K$, and we created special implementations for the 1-way (plain copying) and
2-way (binary merging) cases, while all other $K$-way merges are performed using
the LCP-aware tournament tree.

To make parallel $K$-way LCP-merge more cache- and NUMA transfer-efficient, we
devised a \emph{caching variant}. In \LCPCompare the first character needed for
additional character comparisons during the merge can be predicted (if
comparisons occur at all).  This character is the distinguishing character
between two strings, which we label $\hat{c}_i = s_i[h_i]$, where $h_i =
\lcp_{\Strings}(i)$. Caching this character while sorting is easy, since it is
loaded in a register when the final, distinguishing character comparison is
made. We extended pS$^5$ to save $\hat{c}_i$ in an additional array and employ
it in a modified variant of \LCPCompare to save random accesses across NUMA
nodes.  Using this caching variant all character comparisons accounted for in
the $n \log_2 K + K$ term in Theorem~\ref{thm:kway-LCP-merge} can be done using
the cached $\hat{c}_i$, thus only $\Delta L$ random accesses to strings are
needed for a $K$-way merge.

\subsection{LCP-Insertion Sort}\label{sec:lcp-inssort}

As mentioned in the preceding section, we extended pS$^5$ to save LCP
values. Thus its base-case string sorter, insertion sort, also needed to be
extended. Again, saving and reusing LCPs during insertion sort is a very
intuitive idea, but we found no reference or analysis in the literature.

Assuming the array $\Strings = \arr{ s_1,\ldots,s_{j-1} }$ is already sorted and
the associated LCP array $H$ is known, then insertion sort locates a position
$i$ at which a new string $x = s_j$ can be inserted, while keeping the array
sorted. After insertion, at most the two LCP values $h_i$ and $h_{i+1}$ may need
to be updated. While scanning for the correct position $i$, customarily from the
right, values of both $\Strings$ and $H$ can already be shifted to allocate a
free position.

Using the information in the preliminary LCP array, the scan for $i$ can be
accelerated by skipping over certain areas in which the LCP value attests a
mismatch. The scan corresponds to walking down the LCP interval tree, testing
only one item of each child node, and descending down if it matches. In plainer
words, areas of strings with a common prefix can be identified using the LCP
array (as already mentioned in Section~\ref{sec:merge-kway-parallel}), and it
suffices to \emph{check once} if the candidate matches this common prefix. If
not then the whole area can be skipped.

In the pseudo-code of Algorithm~\ref{alg:lcp-inssort}, the common prefix of the
candidate $x$ is kept in $h'$, and increased while characters match. When a
mismatch occurs, the scan is continued to the left, and all strings can be
skipped if the LCP reveals that the mismatch position remains unchanged
(case~3). If the LCP does below $h'$, then a smaller strings precedes and
therefore the insertion point $i$ is found (case~2). At positions with equal LCP
more characters need to be compared (case~1). In the pseudo-code these three
cases are fused with a copy-loop moving items to right.

\begin{algorithm2e}[t]
\caption{LCP-InsertionSort}\label{alg:lcp-inssort}\normalsize

\KwIn{$\Strings = \{ s_1,\ldots,s_n \}$ is a set of strings with common prefix $\overline{h}$}

\For(\Remi{Insert $x = s_j$ into sorted sequence $\arr{ s_1,\ldots,s_{j-1} }$.}\nllabel{alg:lcp-inssort:for})
{$j = 1,\ldots,n$}
{
  $i := j$,\quad $(x,h') := (s_j,\overline{h})$ \Rem{Start candidate LCP $h'$ with common prefix $\overline{h}$.}
  \While{$i > 1$\nllabel{alg:lcp-inssort:while}}
  {
    \lIf(\Remi{case 1: LCP decreases $\Rightarrow$ insert after $s_{i-1}$.}\nllabel{alg:lcp-inssort:while-begin})
    {$h_i < h'$}
    {
      \KwSty{break}
    }
    \ElseIf(\Remi{case 2: LCP equal $\Rightarrow$ compare more characters.})
    {$h_i = h'$}
    {
      $p := h'$ \Rem{Save LCP of $x$ and $s_i$.}
      \While(\Remi{Compare characters.})
      {$(x[h'] \neq 0 \And x[h'] = s_{i-1}[h'])$\nllabel{alg:lcp-inssort:charcmp1}}
      {
        $h' \Inc$
      }
      \If(\Remi{If $x$ is larger, insert $x$ after $s_{i-1}$,})
      {$x[h'] \geq s_{i-1}[h']$\nllabel{alg:lcp-inssort:charcmp2}}
      {
        $h_i := h'$,\quad $h' := p$ \Rem{set $h_i$, the LCP of $s_{i-1}$ and $x$, but}
        \KwSty{break} \Rem{set $h_{i+1}$ after loop in line~\ref{alg:lcp-inssort:final}.}
      }
    }
    $(s_i, h_{i+1}) := (s_{i-1},h_i)$ \Rem{case 3: LCP is larger $\Rightarrow$ no comparison needed.}
    $i \Dec$ \nllabel{alg:lcp-inssort:while-end}
  }
  $(s_i, h_{i+1}) := (x, h')$ \Rem{Insert $x$ at correct position, update $h_{i+1}$ with LCP.}\nllabel{alg:lcp-inssort:final}\nllabel{alg:lcp-inssort:for-end}
}
  
\KwOut{$\Strings = \arr{ s_1,\ldots,s_n }$ is sorted and has the LCP array $\arr{ \bot,h_2,\ldots,h_n }$}

\end{algorithm2e}

\begin{theorem}
  LCP-aware insertion sort (Algorithm~\ref{alg:lcp-inssort}) requires at most $L
  + \frac{n (n-1)}{2}$ character comparisons and runs in $\Oh{D + n^2}$ time.
\end{theorem}
\begin{proof}
  The only lines containing character comparisons in
  Algorithm~\ref{alg:lcp-inssort} are lines~\ref{alg:lcp-inssort:charcmp1} and
  \ref{alg:lcp-inssort:charcmp2}.  If the while loop condition is \textsl{true},
  then $h'$ is incremented.  In the remaining algorithm the value of $h'$ is
  only shifted around, never discarded or decreased.  Thus we can count the
  number of comparisons yielding a while-loop repetition with $L$.  The while
  loop is encountered at most $\frac{n (n-1)}{2}$, as this is the highest number
  of times the inner loop in lines
  \ref{alg:lcp-inssort:while-begin}--\ref{alg:lcp-inssort:while-end} runs.  We
  can regard the exiting comparison of line \ref{alg:lcp-inssort:charcmp1} and
  the following comparison in line \ref{alg:lcp-inssort:charcmp2} as one ternary
  comparison, as the same values are checked again.  This ternary comparison
  occurs at most once each run of the inner loop, thus $\frac{n (n-1)}{2}$
  times.  With $L \leq D$, the run time follows from the number of iterations of
  the for loop (line \ref{alg:lcp-inssort:for}--\ref{alg:lcp-inssort:for-end})
  and the while loop (lines
  \ref{alg:lcp-inssort:while}--\ref{alg:lcp-inssort:while-end}).
\end{proof}

We close with the remark that non-LCP insertion sort requires $\Oh{D^2 + n^2}$
steps in the worst case, when all strings are equal except for the last
character.

\section{More Shared Memory Parallel String Sorting}\label{sec:more-parasort}

\subsection{Parallel Radix Sort}\label{sec:para-radixsort}

Radix sort is very similar to sample sort, except that classification is much
faster and easier. Hence, we can use the same parallelization toolkit as with
S$^5$. Again, we use three sub-algorithms for differently sized subproblems:
fully parallel radix sort for the original string set and large subsets, a
sequential radix sort for medium-sized subsets and insertion sort for base
cases. Fully parallel radix sort consists of a counting phase, global prefix sum
and a redistribution step. Like in S$^5$, the redistribution is done
out-of-place by copying pointers into a shadow array.

We experimented with 8-bit and 16-bit radixes for the fully parallel
step. Smaller recursive subproblems are processed independently by sequential
radix sort with in-place permuting, and here we found 8-bit radixes to be faster
than 16-bit sorting.  Our parallel radix sort implementation uses the same work
balancing method as parallel S$^5$, freeing the largest subproblems when other
threads are idle.

\subsection{Parallel Caching Multikey Quicksort}\label{sec:para-mkqs}

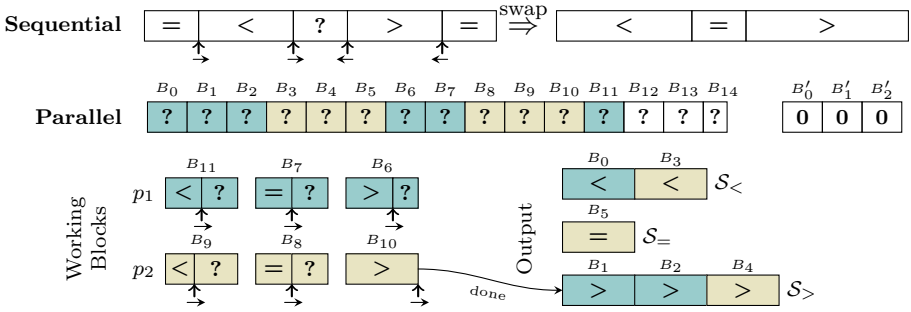
\begin{figure}\centering\small

\begin{tikzpicture}[scale=0.4,
  yscale=-1,
  line cap=round,
  ]

  \colorlet{dred}{black}
  \colorlet{fred}{teal!40!white}

  \colorlet{dgreen}{black}
  \colorlet{fgreen}{olive!20!white}

  \begin{scope}

    \node[left] at (0,0) {\bf Sequential};

    \begin{scope}[xshift=5mm,yshift=-5mm,
      xscale=0.9,
      ptr/.style={<-,thick},
      move/.style={->,>=stealth},
      ]
      \draw (0,0) rectangle (2,1);
      \draw (2,0) rectangle (5.5,1);
      \draw (5.5,0) rectangle (7.5,1);
      \draw (7.5,0) rectangle (11,1);
      \draw (11,0) rectangle (13,1);
      \node at (1,0.5) {$\vphantom{<}$$\boldsymbol{=}$};
      \node at (3.75,0.5) {$\boldsymbol{<}$};
      \node at (6.5,0.5) {$\boldsymbol{?}$};
      \node at (9.25,0.5) {$\boldsymbol{>}$};
      \node at (12,0.5) {$\vphantom{<}$$\boldsymbol{=}$};
      \draw[ptr] (2,1) -- +(0,4mm);
      \draw[ptr] (5.5,1) -- +(0,4mm);
      \draw[ptr] (7.5,1) -- +(0,4mm);
      \draw[ptr] (11,1) -- +(0,4mm);

      \draw[move] (18mm,16mm) -- +(6mm,0);
      \draw[move] (53mm,16mm) -- +(6mm,0);
      \draw[move] (77mm,16mm) -- +(-6mm,0);
      \draw[move] (112mm,16mm) -- +(-6mm,0);
    \end{scope}
    
    \node at (13.1,-0.25) {\large$\overset{\text{swap}}{\Rightarrow}$};

    \begin{scope}[xshift=14.2cm,yshift=-5mm,
      xscale=0.9,
      ]

      \draw (0,0) rectangle (5,1);
      \draw (5,0) rectangle (7,1);
      \draw (7,0) rectangle (13,1);
      \node at (2.5,0.5) {$\boldsymbol{<}$};
      \node at (6,0.5) {$\vphantom{<}$$\boldsymbol{=}$};
      \node at (10,0.5) {$\boldsymbol{>}$};

    \end{scope}

  \end{scope}

  \begin{scope}[yshift=3cm,xscale=1.2,
    blkid/.style={font=\tiny},
    ptr/.style={<-,thick,line cap=round},
    move/.style={->,>=stealth},
    ]

    \node[left] at (0,0) {\bf Parallel};

    \begin{scope}[xshift=5mm,yshift=-5mm,xscale=1.1]
      \foreach \x in {0,...,2,6,7,11} {
        \filldraw[fill=fred] (\x,0) rectangle (\x+1,1);
        \node[font=\color{dred}] at (\x+0.5,0.5) {$\boldsymbol{?}$};
        \node[blkid] at (\x+0.5,-0.4) {$B_{\x}$};
      }
      \foreach \x in {3,...,5,8,9,10} {
        \filldraw[fill=fgreen] (\x,0) rectangle (\x+1,1);
        \node[font=\color{dgreen}] at (\x+0.5,0.5) {$\boldsymbol{?}$};
        \node[blkid] at (\x+0.5,-0.4) {$B_{\x}$};
      }
      \foreach \x in {12,13} {
        \draw (\x,0) rectangle (\x+1,1);
        \node at (\x+0.5,0.5) {$\boldsymbol{?}$};
        \node[blkid] at (\x+0.5,-0.4) {$B_{\x}$};
      }
      \foreach \x in {14} {
        \draw (\x,0) rectangle (\x+0.6,1);
        \node at (\x+0.3,0.5) {$\boldsymbol{?}$};
        \node[blkid] at (\x+0.5,-0.4) {$B_{\x}$};
      }

      \foreach \x/\B in {16/0,17/1,18/2} {
        \draw (\x,0) rectangle (\x+1,1);
        \node at (\x+0.5,0.5) {$\boldsymbol{0}$};
        \node[blkid] at (\x+0.5,-0.4) {$B'_{\B}$};
      }
    \end{scope}

    \node[rotate=90, align=center] at (-1.2,4) {Working \\ Blocks};

    \begin{scope}[xshift=1cm,yshift=2cm]

      \node[dred] at (-0.6,0.6) {$p_1$};

      \filldraw[fill=fred] (0,0) rectangle (1,1);
      \filldraw[fill=fred] (1,0) rectangle (2,1);
      \node[font=\color{dred}] at (0.5,0.5) {$\boldsymbol{<}$};
      \node[font=\color{dred}] at (1.5,0.5) {$\boldsymbol{?}$};
      \draw[ptr] (1,1) -- +(0,4mm);
      \draw[move] (8mm,16mm) -- +(5mm,0);
      \node[blkid] at (1,-0.4) {$B_{11}$};

      \filldraw[fill=fred] (2.5,0) rectangle (3.5,1);
      \filldraw[fill=fred] (3.5,0) rectangle (4.5,1);
      \node[font=\color{dred}] at (3,0.5) {$\vphantom{<}$$\boldsymbol{=}$};
      \node[font=\color{dred}] at (4,0.5) {$\boldsymbol{?}$};
      \draw[ptr] (3.5,1) -- +(0,4mm);
      \draw[move] (33mm,16mm) -- +(5mm,0);
      \node[blkid] at (3.5,-0.4) {$B_{7}$};

      \filldraw[fill=fred] (5,0) rectangle (6.3,1);
      \filldraw[fill=fred] (6.3,0) rectangle (7,1);
      \node[font=\color{dred}] at (5.7,0.5) {$\boldsymbol{>}$};
      \node[font=\color{dred}] at (6.65,0.5) {$\boldsymbol{?}$};
      \draw[ptr] (6.3,1) -- +(0,4mm);
      \draw[move] (61mm,16mm) -- +(5mm,0);
      \node[blkid] at (6,-0.4) {$B_{6}$};

    \end{scope}
    \begin{scope}[xshift=1cm,yshift=45mm]

      \node[dgreen] at (-0.6,0.6) {$p_2$};

      \filldraw[fill=fgreen] (0,0) rectangle (0.8,1);
      \filldraw[fill=fgreen] (0.8,0) rectangle (2,1);
      \node[font=\color{dgreen}] at (0.4,0.5) {$\boldsymbol{<}$};
      \node[font=\color{dgreen}] at (1.4,0.5) {$\boldsymbol{?}$};
      \draw[ptr] (0.8,1) -- +(0,4mm);
      \draw[move] (6mm,16mm) -- +(5mm,0);
      \node[blkid] at (1,-0.4) {$B_{9}$};

      \filldraw[fill=fgreen] (2.5,0) rectangle (3.5,1);
      \filldraw[fill=fgreen] (3.5,0) rectangle (4.5,1);
      \node[font=\color{dgreen}] at (3,0.5) {$\vphantom{<}$$\boldsymbol{=}$};
      \node[font=\color{dgreen}] at (4,0.5) {$\boldsymbol{?}$};
      \draw[ptr] (3.5,1) -- +(0,4mm);
      \draw[move] (33mm,16mm) -- +(5mm,0);
      \node[blkid] at (3.5,-0.4) {$B_{8}$};

      \filldraw[fill=fgreen] (5,0) rectangle (7,1);
      \node[font=\color{dgreen}] at (6,0.5) {$\boldsymbol{>}$};
      \draw[ptr] (7,1) -- +(0,4mm);
      \draw[move] (68mm,16mm) -- +(5mm,0);
      \node[blkid] at (6,-0.4) {$B_{10}$};
      \coordinate (moveA) at (7,0.5);

    \end{scope}

    \node[rotate=90] at (11,4) {Output};

    \begin{scope}[xshift=12cm,yshift=17mm]

      \foreach \x/\c/\C/\B in {0/dred/fred/0,2/dgreen/fgreen/3} {
        \filldraw[fill=\C] (\x+0,0) rectangle (\x+2,1);
        \node[font=\color{\c}] at (\x+1,0.5) {$\boldsymbol{<}$};
        \node[blkid] at (\x+1,-0.3) {$B_{\B}$};
      }

      \node[right] at (4,0.5) {$\mathcal{S}_{<}$};

      \foreach \x/\c/\C/\B in {0/dgreen/fgreen/5} {
        \filldraw[fill=\C] (\x+0,1.75) rectangle (\x+2,2.75);
        \node[font=\color{\c}] at (\x+1,2.25) {$\vphantom{<}$$\boldsymbol{=}$};
        \node[blkid] at (\x+1,1.4) {$B_{\B}$};
      }

      \node[right] at (2,2.25) {$\mathcal{S}_{=}$};

      \foreach \x/\c/\C/\B in {0/dred/fred/1,2/dred/fred/2,4/dgreen/fgreen/4} {
        \filldraw[fill=\C] (\x+0,3.5) rectangle (\x+2,4.5);
        \node[font=\color{\c}] at (\x+1,4) {$\boldsymbol{>}$};
        \node[blkid] at (\x+1,3.2) {$B_{\B}$};
      }

      \node[right] at (6,4) {$\mathcal{S}_{>}$};

      \coordinate (moveB) at (0,4);

    \end{scope}

    \draw[>=stealth,->] (moveA) to[out=0,in=180] node[below,font=\tiny,rotate=-10] {done} (moveB);

  \end{scope}

\end{tikzpicture}
  \caption{Block schema of sequential and parallel multikey quicksort's ternary
    partitioning process}\label{fig:mkqs-blocks}
\end{figure}

Our preliminary experiments with sequential string sorting algorithms (see
Section~\ref{sec:exp-sequential}) showed a surprise winner: an enhanced variant
of multikey quicksort by Tommi Rantala \cite{rantala2007web} often outperformed
more complex algorithms.

This variant employs both caching of characters and uses a super-alphabet of $w
= 8$ characters, exactly as many as fit into a machine word. The string pointer
array is augmented with $w$ cache bytes for each string, and a string subset is
\emph{partitioned by a whole machine word} as splitter. Thereafter, the cached
characters are reused for the recursive subproblems $\Strings_<$ and
$\Strings_>$, and access to strings is needed only for sorting $\Strings_=$,
unless the pivot contains a zero-terminator.  In this section caching means
\emph{copying} of characters into another array, not necessarily into the
processor's cache. Key to the algorithm's good performance is the following
observation:

\begin{theorem}\label{thm:mkqs-access}
  Caching multikey quicksort needs at most $\lfloor \frac{D}{w} \rfloor + n$
  (random) accesses to string characters in total, where $w$ is the number of
  characters cached per access.
\end{theorem}
\begin{proof}
  Per string access $w$ characters are loaded into the cache, and these $w$
  characters are never fetched again. We can thus account for all accesses to
  distinguishing characters using $\lfloor \frac{D}{w} \rfloor$, since the
  characters are fetched in blocks of size $w$. Beyond these, at most one access
  per string can occur, which accounts for fetching $w$ characters of which not
  all are need for sorting.
\end{proof}

In light of this variant's good performance, we designed a parallelized
version. We use three sub-algorithms: \emph{fully parallel caching multikey
  quicksort}, the original sequential caching variant (with explicit recursion
stack) for medium and small subproblems, and insertion sort as base case. For
the fully parallel sub-algorithm, we generalized a block-wise processing
technique from (two-way) parallel atomic quicksort \cite{tsigas2003simple} to
three-way partitioning.

The input array is viewed as a sequence of blocks containing $B$ string pointers
together with their $w$ cache characters (see
Figure~\ref{fig:mkqs-blocks}). Each thread holds exactly three blocks and
performs ternary partitioning by a globally selected pivot. When all items in a
block are classified as $<$, $=$ or $>$, then the block is added to the
corresponding output set $\Strings_<$, $\Strings_=$, or $\Strings_>$. This
continues as long as unpartitioned blocks are available. If no more input blocks
are available, an extra empty memory block is allocated and a second phase
starts. The second partitioning phase ends with fully classified blocks, which
might be only partially filled. Per fully parallel partitioning step there can
be at most $3 p$ partially filled blocks. The output sets $\Strings_<$,
$\Strings_=$, and $\Strings_>$ are processed recursively with threads divided as
evenly among them as possible. The cached characters are updated only for the
$\Strings_=$ set.

In our implementation we use atomic compare-and-swap operations for block-wise
processing of the initial string pointer array and Intel TBB's lock-free queue
for sets of blocks, both as output sets and input sets for recursive steps. When
a partition reaches the threshold for sequential processing, then a continuous
array of string pointers plus cache characters is allocated and the block set is
copied into it. On this continuous array, the usual ternary partitioning scheme
of multikey quicksort is applied sequentially. Like in the other parallelized
algorithms, we use dynamic load balancing and free the largest level when
re-balancing is required. We empirically determined $B = 128\,\text{Ki}$ as a
good block size.

\subsection{Burstsort}\label{sec:para-burstsort}

Burstsort is one of the fastest string sorting algorithms and cache-efficient
for many inputs, but it looks difficult to parallelize it. Keeping a common
burst trie would require prohibitively many synchronized operations, while
building independent burst tries on each PE would lead to the question how to
merge multiple tries of different structure.  This problem of merging tries is
related to parallel $K$-way LCP-merge, and future work may find a way to combine
these approaches.

\section{Experimental Results}\label{sec:experiments}

We implemented parallel versions of S$^5$, $K$-way LCP-merge, multikey quicksort
and radix sort in C++ and compare them with the few parallel string sorting
implementations we could find online in Section~\ref{sec:exp-parallel}.  We also
integrated many sequential implementations into our test framework, and discuss
their performance in Section~\ref{sec:exp-sequential}. Our implementations, the
test framework and most input sets are available from
\url{http://tbingmann.de/2013/parallel-string-sorting}.

\subsection{Experimental Setup}\label{sec:exp-setup}

We tested our implementations and those by other authors on five different
platforms.  All platforms run Linux and their main properties are listed in
Table~\ref{tab:hardware}.  We compiled all programs using gcc~4.6.3 with
optimizations \texttt{-O3 -march=native}.  The five platforms were chosen to
encompass a wide variety of multi-core systems, which exhibit different
characteristics in their memory system and also cover today's most popular
hardware.  By experimenting on a large number of systems (and inputs), we
demonstrate how robust our implementations and algorithm designs are.

The test framework sets up a separate environment for each run.  To isolate heap
fragmentation, it was very important to fork() a child process for each run.
The string data is loaded before the fork(), allocating exactly the matching
amount of RAM, and shared read-only with the child processes. No precaution to
lock the program's memory into RAM was taken (as opposed to previous experiments
reported in~\cite{bingmann2013parallel}). Turbo-mode was disabled on IntelE5.

Before an algorithm is called, the string pointer array is generated inside the
child process by scanning the string data for zero characters, thus flushing
caches and TLB entries. Time measurement is done with clock\_gettime() and
encompasses only the sorting algorithm. Because many algorithms have a deep
recursion stack for our large inputs, we increased the stack size limit to
64\,MiB. For non-NUMA experiments, we took no special precautions of pinning
threads to specific cores or nodes, and used the default Linux task scheduling
system as is. Memory for NUMA-unaware algorithms was interleaved across all
nodes by setting the default allocation policy.

For our experiments with NUMA-aware string sorting, the characters array is
segmented equally onto the NUMA memory banks before running an algorithm.  The
algorithm then pins its threads to the appropriate node, enabling node-local
memory access. Additional allocations are also taken preferably from the local
memory node.

The output of each string sorting algorithm was verified by first checking that
the resulting pointer list is a permutation of the input set, and then checking
that strings are in non-descending order. The input was shared read-only with
the algorithm's process and thus cannot have been modified.

Methodologically we have to discuss, whether measuring only the algorithm's run
time is a good decision. The issue is that deallocation and defragmentation in
both heap allocators and kernel page tables is done lazily. This was most
notable when running two algorithms consecutively. The fork() process isolation
excludes both variables from the experimental results, however, for use in a
real program context these costs cannot be ignored. We currently do not know
how to invoke the lazy cleanup procedures to regenerate a pristine memory
environment. These issues must be discussed in greater detail in future work for
sound results with big data in RAM.  We briefly considered HugePages, but these
did not yield a performance boost. This is probably due to random accesses being
the main time cost of string sorting, while the number of TLB entries is not a
bottleneck.

\begin{table}\centering\small\def\tabcolsep{3.5pt}
\caption{Hard- and software characteristics of experimental platforms}\label{tab:hardware}
\begin{tabular}{l|l|r|r||r|r|r|r}
Name    & Processor          & Clock & Sockets $\times$                     & Cache: L1      & L2              & L3            & RAM   \\
        &                    & [GHz] & Cores $\times$ HT                    & [KiB]          & [KiB]           & [MiB]         & [GiB] \\\hline
IntelE5 & Intel Xeon E5-4640 & 2.4   & $4 \times 8 \times 2$                & $32 \times 32$ & $32 \times 256$ & $4 \times 20$ & 512   \\
AMD48   & AMD Opteron 6168   & 1.9   & $4 \times 12 \:\phantom{\:\times 1}$ & $48 \times 64$ & $48 \times 512$ & $8 \times 6$  & 256   \\
AMD16   & AMD Opteron 8350   & 2.0   & $4 \times 4 \:\phantom{\:\times 1}$  & $16 \times 64$ & $16 \times 512$ & $4 \times 2$  & 64    \\
Inteli7 & Intel Core i7 920  & 2.67  & $1 \times 4 \times 2$                & $4 \times 32$  & $4 \times 256$  & $1 \times 8$  & 12    \\
IntelX5 & Intel Xeon X5355   & 2.67  & $2 \times 4 \times 1$                & $8 \times 32$  & $4 \times 4096$ &               & 16    \\ 
\end{tabular}

\bigskip
\def\tabcolsep{2.0pt}
\begin{tabular}{l|l|l|c|l|l}
Name    & Codename     & Memory               & NUMA  & Interconnect            & Linux/Kernel Version \\
        &              & Channels             & Nodes &                         &                      \\ \hline
IntelE5 & Sandy Bridge & 4 $\times$ DDR3-1600 & 4     & 2 $\times$ 8.0 GT/s QPI & Ubuntu 12.04/3.2.0   \\
AMD48   & Magny-Cours  & 4 $\times$ DDR3-667  & 8     & 4 $\times$ 3.2 GHz HT   & Ubuntu 12.04/3.2.0   \\
AMD16   & Barcelona    & 2 $\times$ DDR2-533  & 4     & 3 $\times$ 1.0 GHz HT   & Ubuntu 10.04/2.6.32  \\
Inteli7 & Bloomfield   & 3 $\times$ DDR3-800  &       & 1 $\times$ 4.8 GT/s QPI & openSUSE 11.3/2.6.34 \\
IntelX5 & Clovertown   & 2 $\times$ DDR2-667  &       & 1 $\times$ 1.3 GHz FSB  & Ubuntu 12.04/3.2.0   \\
\end{tabular}
\end{table}

\subsection{Inputs}\label{sec:exp-inputs}

We selected the following datasets, all with 8-bit characters. Most important
characteristics of these instances are shown in Table~\ref{tab:data}.

\textbf{URLs} contains all URLs found on a set of web pages which were crawled
breadth-first from the author's institute website. They include the protocol
name.

\textbf{Random} (from \cite{sinha2004cache-conscious}) are strings of length
$[0 \mathop{:} 20)$ over the ASCII alphabet $[33 \mathop{:} 127)$, with both length and characters
chosen uniformly random.

\textbf{GOV2} is a TREC test collection consisting of 25 million HTML pages, PDF
and other documents retrieved from websites under the .gov domain. We
consider the whole corpus for line-based string sorting, concatenated by
document id.

\textbf{Wikipedia} is an XML dump of the most recent version of all pages in the
English Wikipedia, which was obtained from \url{http://dumps.wikimedia.org/};
our dump is dated \texttt{enwiki-20120601}.  Since the XML data is not
line-based, we perform \emph{suffix sorting} on this input.

We also include the three largest inputs Ranjan \textbf{Sinha}
\cite{sinha2004cache-conscious} tested burstsort on: a set of \textbf{URLs}
excluding the protocol name, a sequence of genomic strings of length 9 over a
\textbf{DNA} alphabet, and a list of non-duplicate English words called
\textbf{NoDup}. The ``largest'' among these is NoDup with only 382\,MiB, which
is why we consider these inputs more as reference datasets than as our target.

The inputs were chosen to represent both real-world datasets, and to exhibit
extreme results when sorting. Random has a very low average LCP, while URLs have
a high average LCP. GOV2 is a general text file with all possible ASCII
characters, and Sinha's DNA has a small alphabet size. By taking suffixes of
Wikipedia we have a very large sorting problem instance, which needs little
memory for characters.

Our inputs are very large, one infinite, and most of our platforms did not have
enough RAM to process them. For each platform, we determined a large prefix
$[0 \mathop{:} n)$, which can be processed with the available RAM and time, and leave
sorting of the remainder to future work.

\begin{table}[tb]\centering\normalsize
\caption{Characteristics of the selected input instances.}\label{tab:data}
\def\tabcolsep{6pt}
\begin{tabular}{l|rrrrrr}
Name        & $n$      & $N$                       & $\frac{D}{N}$ ($D$) & $\frac{L}{n}$ & $|\Sigma|$ & avg.\ $|s|$             \\ \hline
URLs        & 1.11\,G  & 70.7\,Gi                  & 93.5\,\%            & 62.0          & 84         & 68.4                    \\
Random      & $\infty$ & $\infty$                  & $-$                 & $-$           & 94         & 10.5                    \\
GOV2        & 11.3\,G  & 425\,Gi                   & 84.7\,\%            & 32.0          & 255        & 40.3                    \\
Wikipedia   & 83.3\,G  & $\frac{1}{2} n (n\!+\!1)$ & (79.56\,T)          & 954.7         & 213        & $\frac{1}{2} (n\!+\!1)$ \\
Sinha URLs  & 10\,M    & 304\,Mi                   & 97.5\,\%            & 29.4          & 114        & 31.9                    \\
Sinha DNA   & 31.6\,M  & 302\,Mi                   & 100\,\%             & 9.0           & 4          & 10.0                    \\
Sinha NoDup & 31.6\,M  & 382\,Mi                   & 73.4\,\%            & 7.7           & 62         & 12.7                    \\
\end{tabular}
\end{table}

\subsection{Performance of Parallel Algorithms}\label{sec:exp-parallel}

In this section we report on our experiments on the platforms shown in
Table~\ref{tab:hardware}, which contains a wide variety of multi-core machines
of different age. The results plotted in
Figures~\ref{fig:more-IntelE5}--\ref{fig:more-IntelX5} show the speed up of each
parallel algorithm over the best sequential one, for increasing thread count.
Tables \ref{tab:absrun-IntelE5}--\ref{tab:absrun-IntelX5b} show absolute running
times of our experiments, with the fastest algorithm's time highlighted in bold
text.

Overall, our parallel string sorting implementations yield high speedups, which
are generally much higher than those of all previously existing parallel string
sorters. Each individual parallel algorithm's speedup depends highly on hardware
characteristics like processor speed, RAM and cache performance\footnote{See
  \url{http://tbingmann.de/2013/pmbw/} for parallel memory bandwidth
  experiments}, the interconnection between sockets, and the input's
characteristics. In general, the speedup of string sorting for high thread
counts is bounded by memory bandwidth, not processing power.  On both non-NUMA
platforms (Figures~\ref{fig:more-Inteli7}, \ref{fig:more-IntelX5}), our
implementations of pS$^5$ are the string sorting algorithm with highest
speedups, except for Random and Sinha's NoDup inputs.  On NUMA many-core
platforms, the picture is more complex and results mostly depend on how well the
inner loops and memory transfers are optimized on each particular system.

The parallel experiments cover all algorithms we describe in this paper:
pS$^5$-Unroll is a variant of pS$^5$ from Section~\ref{sec:s5}, which
interleaves three unrolled descents of the classification tree, while
pS$^5$-Equal unrolls only a single descent, but tests equality at each splitter
node. In the NUMA-aware variant called ``pS$^5$-Unroll + pLCP-Merge'' we first
run pS$^5$-Unroll independently on each NUMA node for separate parts of the
input, and then merge the presorted parts using our parallel $K$-way LCP-merge
algorithm (Section~\ref{sec:para-mergesort}). From the additional parallel
algorithms in Section~\ref{sec:more-parasort}, we draw our parallel multikey
quicksort (pMKQS) implementations, and radix sorts with 8-bit and 16-bit fully
parallel steps.  Furthermore, we included the parallel radix sort implemented by
Ta\-kuya Akiba~\cite{akiba2011radixsort} in the experiments on all platforms.

For the tests on Inteli7 and IntelX5, we added three more parallel
implementations: pMKQS-SIMD is a multikey quicksort implementation from
Rantala's library, which uses SIMD instructions to perform vectorized
classification against a single pivot. We improved the code to use OpenMP tasks
for recursive sorting steps. The second implementation is a parallel 2-way
LCP-mergesort also by Rantala, which we also augmented with OpenMP
tasks. However, only recursive merges are run in parallel, the largest merge is
performed sequentially. The implementation uses insertion sort for $|\Strings| <
32$, all other sorting is done via merging. N.\ Shamsundar's parallel
LCP-mergesort is the third additional implementation, but it also uses only
2-way merges.  As seen in Figures~\ref{fig:more-Inteli7}--\ref{fig:more-IntelX5},
only Akiba's radix sort scales fairly well, which is why we omitted the other
three algorithms on the platforms with more than eight cores.

\emph{Inteli7} (Figure~\ref{fig:more-Inteli7}, Table~\ref{tab:absrun-Inteli7}--\ref{tab:absrun-Inteli7b}) is a consumer-grade, single
socket machine with fast RAM and cache hierarchy. \emph{IntelX5}
(Figure~\ref{fig:more-IntelX5}, Table~\ref{tab:absrun-IntelX5}--\ref{tab:absrun-IntelX5b}) is our oldest architecture, and shows the
slowest absolute speedups. Both are not NUMA architectures, which is why we did
not run our NUMA-aware algorithm on them. They are more classic architectures,
and exhibit most of the effects we targeted in our algorithms to gain good
speedups. Our pS$^5$ variants are fastest on all inputs, except very random ones
(Random and NoDup), where radix sorts are slightly faster on Inteli7.
Remarkably, on IntelX5 the speed gain of radix sort is not visible.  We suspect
that some processors can optimize the inner loops of radix sort (counting,
prefix sums and data redistributions with few input/output streams
\cite{karkkainen2009engineering}) better than others.  Our pMKQS also shows good
overall speedups, but is never particularly fast.  This is due to the high
memory bandwidth caching multikey quicksort requires, as it reads and rereads an
array to just partition by one pivot.

For all test instances, except URLs, the fully parallel sub-algorithm of pS$^5$
was run only 1--4 times. Thereafter, the input was split up into enough subsets,
and most of the additional speedup is gained by load-balancing the sequential
sub-algorithms well. The pS$^5$-Equal variant handles URL instances better, as
many equal matches occur here. However, for all other inputs, pS$^5$-Unroll with
interleaved tree descents fares better, even though it has higher theoretical
running time.

Comparing our radix sorts with Akiba's we already see the implementation's main
problems: it does not parallelize recursive sorting steps (only the top-level is
parallelized) and only performs simple load balancing. This can be seen most
pronounced on URLs and GOV2.  All three additional implementations, pMKQS-SIMD,
pMergesort-2way by Rantala, and the same by Shamsundar do not show any good
speedup, partly because they are already pretty slow sequentially, and partly
because they are not fully parallelized.

On the Inteli7 machine, which has four real cores and four Hyper-Threading
cores, pS$^5$ achieves speedups $\geq 3.2$ on all inputs, except Random where it
gains only $2.5$. This is remarkable, as the machine has only three memory
channels, and a single core can fully utilize two of them. Thus in pS$^5$ a lot
of computation work is parallelized. On IntelX5, which has eight real cores,
pS$^5$ achieves speedups $\geq 3$ on all inputs. We attribute this to the early
dual-socket architecture, on which many other parallel implementations also do
not scale well.

\emph{IntelE5} (Figure~\ref{fig:more-IntelE5}, Table~\ref{tab:absrun-IntelE5}) is our newest machine with 32
real cores across four sockets with one NUMA node each. It contains one of
Intel's most recent many-core processors. \emph{AMD48}
(Figure~\ref{fig:more-AMD48}, Table~\ref{tab:absrun-AMD48}) is a somewhat older AMD many-core machine with
high core count, but relatively slow RAM and a slower interconnect.  Compared to
the previous results on Inteli7, we notice that parallel multikey quicksort (pMKQS) is very
fast, and achieves slightly higher speedups than pS$^5$ on most inputs on
IntelE5 and significantly higher ones on AMD48.  This effect is clearly due to
pS$^5$ ignoring the NUMA architecture and thus incurring a relatively large
number of expensive inter-node random string accesses. We analyzed the number of
string access of pMKQS in Theorem~\ref{thm:mkqs-access}, after which the
characters are saved and accessed in a scanning pattern. This scanning
apparently works well on the NUMA machines, as it is very cache-efficient, can
be easily predicted by the processor's memory prefetcher, and a costly
inter-node transfered cache line contains saved characters of eight strings.

These expected results were the reason to focus on NUMA-aware string sorting
algorithms, and to develop parallel $K$-way LCP-merge for top-level merging of
presorted sequences. In our experiments we ran ``pS$^5$-Unroll + pLCP-Merge''
only when there is at least one thread per NUMA node. We tried to rebalance
threads to other NUMA nodes once work on a node is done, but this did not work
well, since the additional inter-node synchronization was too costly. We thus
have to leave the question of how to balance sorting work on NUMA systems for
highly skewed inputs open to future research.  Plain LCP-merge also contains
costly inter-node random string accesses in case 1 of \LCPCompare. As predicted
in Section~\ref{sec:merge-details}, we saw a huge speed improvement due to
\emph{caching} of just the distinguishing character $\hat{c}$, and don't
consider the non-caching variant in our results.

On IntelE5, with four NUMA nodes, pS$^5$-Unroll + pLCP-Merge reaches the highest
speedups on URLs, GOV2 inputs and Wikipedia suffixes. On the AMD48 machine with
eight NUMA nodes, random access is even more costly and the inter-node
connections are easily congested, which is why pMKQS fairs better against our
NUMA-aware sorting.  In future (possibly the next revision of this
paper), experiments with caching more than just one character may lead to larger
speedups on these NUMA systems. Remarkably, radix sort is still very fast on
both NUMA machines for random inputs.

The lower three plots in Figure~\ref{fig:more-IntelE5} and \ref{fig:more-AMD48}
show that on these large many-core platforms, parallel sorting becomes less
efficient for small inputs (around 300\,MiB). This is expected due to the high
cost of synchronization, but our parallel algorithms still fare well.

\emph{AMD16} (Figure~\ref{fig:more-AMD16}, Table~\ref{tab:absrun-AMD16}) is an earlier NUMA architecture with
four NUMA nodes, and the slowest RAM speed and interconnect in our
experiment. However, on this machine random access, memory bandwidth and
processing power (in cache) seems to be more balanced for pS$^5$ than on the
newer NUMA machines.

We included the absolute running times of all our speedup experiments in
Tables~\ref{tab:absrun-IntelE5}--\ref{tab:absrun-IntelX5b} for reference and to
show that our parallel implementations scale well both for very large instances
on many-core platforms and also for small inputs on machines with fewer cores.
For a final version of this article, we may need to remove the data
table. In this case, we will reference a technical report instead.

\subsection{Performance of Sequential Algorithms}\label{sec:exp-sequential}

We collected many sequential string sorting algorithms in our test framework. We
believe it to contain virtually every string sorting implementation publicly
available.

The algorithm library by Tommi Rantala \cite{rantala2007web} contains 37~versions
of radix sort (in-place, out-of-place, and one-pass with various dynamic memory
allocation schemes), 26~variants of multikey quicksort (with caching,
block-based, different dynamic memory allocation and SIMD instructions),
10~different funnelsorts, 38~implementations of burstsort (again with different
dynamic memory managements), and 29~mergesorts (with losertree and LCP caching
variants). In total these are 140~original implementation variants, all of high
quality.

The other main source of string sorting implementations are the publications of
Ranjan Sinha. We included the original burstsort implementations (one with
dynamically growing arrays and one with linked lists), and 9~versions of
copy-burstsort. The original copy-burstsort code was written for 32-bit
machines, and we modified it to work with 64-bit pointers.

We also incorporated the implementations of CRadix sort and LCP-Merge\-sort by
Waihong Ng, and the original multikey quicksort code by Bentley and Sedgewick.

\begin{table}[tb]\centering\normalsize
\caption{Description of selected sequential string sorting algorithms}\label{tab:seqalgo}
\begin{tabularx}{\linewidth}{l|X}
Name             & Description and Author                                                                                                                                        \\ \hline
std::sort        & \texttt{gcc}'s standard atomic introsort with full string comparisons.                                                                                        \\
mkqs             & Original multikey quicksort by Bentley and Sedgewick \cite{bentley1997fast}.                                                                                  \\
mkqs\_cache8     & Modified multikey quicksort with caching of eight characters by Tommi Rantala \cite{rantala2007web}, slightly improved.                                       \\
radix8\_CI       & 8-bit in-place radix sort by Tommi Rantala \cite{karkkainen2009engineering}.                                                                                  \\
radix16\_CI      & Adaptive 16-/8-bit in-place radix sort by Tommi Rantala \cite{karkkainen2009engineering}.                                                                     \\
radixR\_CE7      & Adaptive 16-/8-bit out-of-place radix sort by Tommi Rantala~\cite{karkkainen2009engineering}, version CE7 (preallocated swap array, unrolling, sorted-check). \\
CRadix           & Cache efficient radix sort by Waihong Ng \cite{ng2007cache}, unmodified.                                                                                      \\
LCPMergesort     & LCP-mergesort by Waihong Ng \cite{ng2008merging}, unmodified.                                                                                                 \\
Seq-S$^5$-Unroll & Sequential super scalar string sample sort with interleaved loop over strings, unrolled tree traversal and radix sort as base sorter.                         \\
Seq-S$^5$-Equal  & Sequential super scalar string sample sort with equality check, unrolled tree traversal and radix sort as base sorter.                                        \\
burstsortA       & Burstsort using dynamic arrays by Ranjan Sinha \cite{sinha2004cache-conscious}, from \cite{rantala2007web}.                                                   \\
fbC-burstsort    & Copy-Burstsort with ``free bursts'' by Ranjan Sinha \cite{sinha2007cache-efficient}, heavily repaired and modified to work with 64-bit pointers.              \\
sCPL-burstsort   & Copy-Burstsort with sampling, pointers and only limited copying depth by Ranjan Sinha \cite{sinha2007cache-efficient}, also heavily repaired.                 \\
\end{tabularx}
\end{table}

Of the 203 different sequential string sorting variants, we selected the
thirteen implementations listed in Table~\ref{tab:seqalgo} to represent both the
fastest ones in a preliminary test and each of the basic algorithms from
Section~\ref{sec:basic-sequential}. The thirteen algorithms were run on all our
five test platforms on small portions of the test instances described in
Section~\ref{sec:experiments}. Tables \ref{tab:seqalgo-results1} and
\ref{tab:seqalgo-results2} show the results, with the fastest algorithm's time
highlighted with bold text.

Cells in the tables without value indicate a program error, out-of-memory
exceptions or extremely long runtime. This was always the case for the
copy-burstsort variants on the GOV2 and Wikipedia inputs, because they perform
excessive caching of characters. On Inteli7, some implementations required more
memory than the available 12\,GiB to sort the 4\,GiB prefixes of Random and URLs.

Over all run instances and platforms, multikey quicksort with caching of eight
characters was fastest on 18~pairs, winning the most tests. It was fastest on
all platforms for both URL list and GOV2 prefixes, except URL on IntelX5, and on
all large instances on AMD48 and AMD16. However, for the NoDup input, short
strings with large alphabet, the highly tuned radix sort radixR\_CE7
consistently outperformed mkqs\_cache8 on all platforms by a small margin. The
copy-burstsort variant fbC\_burstsort was most efficient on all platforms for
DNA, which are short strings with small alphabet. For Random strings and
Wikipedia suffixes, mkqs\_cache8 or radixR\_CE7 was fastest, depending on the
platforms memory bandwidth and sequential processing speed.  Our own
\emph{sequential} implementations of S$^5$ were never the fastest, but they
consistently fall in the middle field, without any outliers.  This is expected,
since S$^5$ is mainly designed to be used as an efficient top-level parallel
algorithm, and to be conservative with memory bandwidth, since this is the
limiting factor for data-intensive multi-core applications.

We also measured the peak memory usage of the sequential implementations using a
heap and stack profiling
tool\footnote{\url{http://tbingmann.de/2013/malloc_count/}, by one of the
  authors.} for the selected sequential test instances. The bottom of
Table~\ref{tab:seqalgo-results1} shows the results in MiB, excluding the string
data array and the string pointer array (we only have 64-bit systems, so
pointers are eight bytes). We must note that the profiler considers
\emph{allocated virtual memory}, which may not be identical to the amount of
physical memory actually used. From the table we plainly see, that the more
\emph{caching} an implementation does, the higher its peak memory
allocation. However, the memory usage of fbC\_burstsort is extreme, even if one
considers that the implementation can deallocate and recreate the string data
from the burst trie. The lower memory usage of fbC\_burstsort for Random is due
to the high percentage of characters stored implicitly in the trie structure.
The sCPL\_burstsort and burstsortA variants bring the memory requirement down
somewhat, but they are still high. Some radixsort variants and, most notable,
mkqs\_cache8 are also not particularly memory conservative, again due to
caching. Our sequential S$^5$ implementation fares well in this comparison
because it does no caching and permutes the string pointers in-place (Note that
radixsort is used for small string subsets in sequential S$^5$. This is due to
the development history: we finished sequential S$^5$ before focusing on caching
multikey quicksort). For sorting with little extra memory, plain multikey
quicksort is still a good choice.



\begin{table}[p]\centering\small
\caption{Run time of sequential algorithms on IntelE5 and AMD48 in seconds, and peak memory usage of algorithms on IntelE5. See Table~\ref{tab:seqalgo} for a short description of each.}\label{tab:seqalgo-results1}
\def\tabcolsep{6pt}
\begin{tabular}{l|rrrr*{3}{r}|}
                    & \multicolumn{4}{c|}{Our Datasets} & \multicolumn{3}{c|}{Sinha's} \\
                    & URLs     & Random   & GOV2     & \multicolumn{1}{r|}{Wikipedia} & URLs     & DNA     & NoDup    \\ \hline
$n$                 & 66\,M    & 409\,M   & 80.2\,M  & 256\,Mi   & 10\,M    & 31.5\,M & 31.6\,M  \\
$N$                 & 4\,Gi    & 4\,Gi    & 4\,Gi    & 32\,Pi    & 304\,Mi  & 302\,Mi & 382\,Mi  \\
$D / N$ ($D$)       & 92.7\,\% & 43.0\,\% & 69.7\,\% & (13.6\,G) & 97.5\,\% & 100\,\% & 73.4\,\% \\
$L / n$             & 57.9     & 3.3      & 34.1     & 33.0      & 29.4     & 9.0     & 7.7      \\ \hline
                    & \multicolumn{7}{c|}{IntelE5} \\ \cline{2-8}
       std::sort &      122 &      422 &      153 &      287 &     11.6 &     26.9 &     25.4 \\
            mkqs &     37.1 &      228 &     56.7 &      129 &     5.67 &     11.0 &     10.9 \\
    mkqs\_cache8 & \bf 16.6 &     67.1 & \bf 25.7 &     79.5 & \bf 2.03 &     4.62 &     6.02 \\
      radix8\_CI &     48.4 &     64.3 &     54.6 &     90.1 &     6.12 &     6.79 &     6.29 \\
     radixR\_CE7 &     37.5 & \bf 58.6 &     44.6 & \bf 72.4 &     4.77 &     4.66 & \bf 4.73 \\
Seq-S$^5$-Unroll &     32.4 &      142 &     39.9 &      103 &     4.90 &     7.05 &     7.78 \\
 Seq-S$^5$-Equal &     32.8 &      169 &     45.6 &      120 &     5.11 &     7.68 &     8.38 \\
          CRadix &     54.1 &     65.1 &     61.6 &      113 &     6.82 &     10.2 &     8.63 \\
    LCPMergesort &     25.8 &      316 &     53.9 &      167 &     5.00 &     14.6 &     17.0 \\
      burstsortA &     29.5 &      131 &     43.3 &      120 &     5.64 &     8.48 &     8.46 \\
  fbC\_burstsort &     60.9 &     69.8 &          &          &     11.0 & \bf 3.81 &     15.4 \\
 sCPL\_burstsort &     45.0 &      122 &          &          &     11.1 &     14.3 &     24.8 \\ \hline
                 & \multicolumn{7}{c|}{AMD48} \\ \cline{2-8}
       std::sort &      243 &   1\,071 &      197 &     494 &     21.7 &     55.7 &     44.2 \\
            mkqs &     90.7 &      511 &     78.2 &     226 &     11.0 &     20.8 &     19.8 \\
    mkqs\_cache8 & \bf 37.9 & \bf 96.9 & \bf 31.7 & \bf 114 & \bf 3.44 &     7.08 &     8.75 \\
      radix8\_CI &     83.1 &      127 &     71.2 &     138 &     9.72 &     10.9 &     9.46 \\
     radixR\_CE7 &     73.6 &      125 &     63.6 &     120 &     8.34 &     8.27 & \bf 7.70 \\
Seq-S$^5$-Unroll &     59.4 &      283 &     55.6 &     167 &     7.93 &     11.2 &     11.7 \\
 Seq-S$^5$-Equal &     56.5 &      292 &     57.7 &     180 &     8.06 &     11.5 &     12.2 \\
          CRadix &     98.2 &      115 &     68.8 &     147 &     8.11 &     12.6 &     11.2 \\
    LCPMergesort &     48.2 &      597 &     68.8 &     232 &     7.35 &     20.7 &     24.4 \\
      burstsortA &     46.0 &      214 &     53.0 &     193 &     8.49 &     13.3 &     13.1 \\
  fbC\_burstsort &     85.8 &      115 &          &         &     17.7 & \bf 5.92 &     22.1 \\
 sCPL\_burstsort &     73.1 &      266 &          &         &     20.1 &     24.6 &     37.9 \\ \hline
                 & \multicolumn{7}{p{76ex}|}{\centering{}Memory usage of sequential algorithms (on IntelE5) in MiB, excluding input and string pointer array} \\ \cline{2-8}
       std::sort &   0.002 &  0.002 &  0.002 &  0.003 &  0.002 &  0.002 &  0.002 \\
            mkqs &   0.134 &  0.003 &   1.66 &  0.141 &  0.015 &  0.003 &  0.004 \\
    mkqs\_cache8 &  1\,002 & 6\,242 & 1\,225 & 4\,096 &    153 &    483 &    483 \\
      radix8\_CI &    62.7 &    390 &   77.2 &    256 &   9.55 &   30.2 &   30.2 \\
     radixR\_CE7 &     669 & 3\,902 &    786 & 2\,567 &    111 &    303 &    303 \\
Seq-S$^5$-Unroll &     129 &    781 &    155 &    513 &   20.3 &   60.8 &   60.9 \\
 Seq-S$^5$-Equal &     131 &    781 &    156 &    513 &   20.8 &   60.8 &   61.0 \\
          CRadix &     752 & 4\,681 &    919 & 3\,072 &    114 &    362 &    362 \\
    LCPMergesort &  1\,002 & 6\,242 & 1\,225 & 4\,096 &    153 &    483 &    483 \\
      burstsortA &  1\,466 & 7\,384 & 1\,437 & 5\,809 &    200 &    531 &    792 \\
  fbC\_burstsort & 31\,962 & 6\,200 &        &        & 2\,875 &    436 & 4\,182 \\
 sCPL\_burstsort &  9\,971 & 7\,262 &        &        & 1\,578 & 1\,697 & 5\,830 \\ \hline
\end{tabular}
\end{table}

\begin{table}[p]\centering\small
\caption{Run time of sequential algorithms on AMD16, Inteli7, and IntelX5 in seconds. See Table~\ref{tab:seqalgo} for a short description of each.}\label{tab:seqalgo-results2}
\def\tabcolsep{6pt}
\begin{tabular}{l|rrrr rrr|}
                    & \multicolumn{4}{c|}{Our Datasets} & \multicolumn{3}{c|}{Sinha's} \\
                    & URLs     & Random   & GOV2     & \multicolumn{1}{r|}{Wikipedia} & URLs     & DNA     & NoDup    \\ \hline
$n$                 & 66\,M    & 409\,M   & 80.2\,M  & 256\,Mi   & 10\,M        & 31.5\,M & 31.6\,M  \\
$N$                 & 4\,Gi    & 4\,Gi    & 4\,Gi    & 32\,Pi    & 304\,Mi      & 302\,Mi & 382\,Mi  \\
$D / N$ ($D$)       & 92.6\,\% & 43.0\,\% & 69.7\,\% & (13.6\,G) & 97.5\,\%     & 100\,\% & 73.4\,\% \\
$L / n$             & 57.9     & 3.3      & 34.1     & 33.0      & 29.4         & 9.0     & 7.7      \\ \hline
                    & \multicolumn{7}{c|}{AMD16} \\ \cline{2-8}
       std::sort &      274 &  1\,088 &      237 &         &     28.1 &     73.5 &     56.8 \\
            mkqs &      138 &     586 &     99.7 &     284 &     15.9 &     29.6 &     26.9 \\
    mkqs\_cache8 & \bf 45.4 & \bf 114 & \bf 40.2 &     142 & \bf 4.77 &     8.99 &     10.7 \\
      radix8\_CI &      112 &     158 &     84.5 &     171 &     11.7 &     13.6 &     11.5 \\
     radixR\_CE7 &     91.4 &     156 &     75.3 & \bf 135 &     10.5 &     10.8 & \bf 9.46 \\
Seq-S$^5$-Unroll &     70.6 &     326 &     68.4 &     235 &     9.57 &     14.5 &     14.8 \\
 Seq-S$^5$-Equal &     72.9 &     315 &     67.9 &     227 &     9.41 &     12.8 &     13.5 \\
          CRadix &      132 &     128 &     91.8 &     201 &     11.7 &     19.0 &     14.7 \\
    LCPMergesort &     56.8 &     631 &     86.6 &     285 &     9.22 &     25.5 &     30.9 \\
      burstsortA &     52.4 &     284 &     64.4 &     252 &     10.1 &     17.2 &     17.0 \\
  fbC\_burstsort &      109 &     123 &          &         &     25.5 & \bf 6.34 &     29.8 \\
 sCPL\_burstsort &     79.8 &     288 &          &         &     28.8 &     38.9 &     54.9 \\ \hline
                 & \multicolumn{7}{c|}{Inteli7} \\ \cline{2-8}
       std::sort &     94.4 &      360 &     75.6 &      233 &     9.41 &     23.3 &     21.2 \\
            mkqs &     33.2 &      187 &     30.9 &      112 &     5.00 &     9.43 &     9.62 \\
    mkqs\_cache8 & \bf 14.7 &          & \bf 14.5 &     66.1 & \bf 1.81 &     3.91 &     5.10 \\
      radix8\_CI &     38.3 &     49.6 &     32.4 &     73.0 &     4.95 &     5.35 &     5.05 \\
     radixR\_CE7 &     30.7 & \bf 46.6 &     28.8 & \bf 59.7 &     3.85 &     3.71 & \bf 3.84 \\
Seq-S$^5$-Unroll &     25.3 &      108 &     26.2 &     86.5 &     3.94 &     5.75 &     6.44 \\
 Seq-S$^5$-Equal &     25.9 &      130 &     27.6 &     97.4 &     4.11 &     6.14 &     6.82 \\
          CRadix &     43.2 &          &     33.4 &     84.6 &     5.27 &     7.87 &     6.49 \\
    LCPMergesort &     22.7 &          &     32.3 &      139 &     4.33 &     12.1 &     14.2 \\
      burstsortA &     22.5 &          &     24.9 &      102 &     4.52 &     6.67 &     6.91 \\
  fbC\_burstsort &          &          &          &          &     9.31 & \bf 3.12 &     12.7 \\
 sCPL\_burstsort &     35.3 &          &          &          &     9.82 &     13.0 &     20.0 \\ \hline
                 & \multicolumn{7}{c|}{IntelX5} \\ \cline{2-8}
       std::sort &      140 &      731 &      137 &      401 &     17.4 &     48.9 &     38.6 \\
            mkqs &     74.2 &      333 &     56.6 &      148 &     7.18 &     13.9 &     14.0 \\
    mkqs\_cache8 &     30.1 & \bf 80.1 & \bf 25.5 &     95.5 & \bf 3.35 &     6.48 &     7.51 \\
      radix8\_CI &     69.0 &      109 &     49.7 &     88.9 &     5.84 &     7.47 &     7.06 \\
     radixR\_CE7 &     55.9 &      110 &     44.6 & \bf 83.5 &     4.92 &     6.12 & \bf 5.95 \\
Seq-S$^5$-Unroll &     35.9 &      170 &     34.4 &      108 &     4.17 &     5.62 &     6.74 \\
 Seq-S$^5$-Equal &     38.7 &      198 &     36.5 &      121 &     4.65 &     6.09 &     7.19 \\
          CRadix &     77.7 &     94.3 &     57.7 &      141 &     8.26 &     13.3 &     11.0 \\
    LCPMergesort &     36.2 &      454 &     55.1 &      208 &     6.61 &     19.3 &     22.8 \\
      burstsortA & \bf 27.0 &      215 &     34.9 &      153 &     4.70 &     9.17 &     10.1 \\
  fbC\_burstsort &          &     86.8 &          &          &     16.4 & \bf 4.31 &     21.1 \\
 sCPL\_burstsort &     46.5 &      212 &          &          &     19.6 &     26.8 &     38.5 \\ \hline
\end{tabular}
\end{table}

\input{speedup-plots.tex}


\begin{table}\centering\small
\caption{Absolute run time of parallel and best sequential algorithms on IntelE5 in seconds, median of 1--3 runs. See Table~\ref{tab:paraalgo} for a short description of each.}\label{tab:absrun-IntelE5}
\def\tabcolsep{3.6pt}
\begin{tabular}{l|*{10}{r}|@{}}
PEs & 1   & 2 & 4 & 8 & 12 & 16 & 24 & 32 & 48 & 64 \\ \hline
& \multicolumn{10}{l|}{\textbf{URLs} (complete), $n = 1.11\,\text{G}$, $N = 70.7\,\text{Gi}$, $\frac{D}{N} = 93.5\,\%$} \\ \cline{2-11}
mkqs\_cache8 & \bf 467 &  &  &  &  &  &  &  &  &  \\
pS$^5$-Unroll &    633 &     310 &     156 &     92.8 &     69.2 &     52.7 &     45.1 &     42.8 &     41.1 &     39.9 \\
 pS$^5$-Equal &    646 &     316 &     157 &     93.0 &     69.3 &     52.8 &     45.8 &     42.0 &     41.2 &     41.2 \\
 pS$^5$+LCP-M &        &         & \bf 116 & \bf 74.4 &     57.7 & \bf 47.0 & \bf 38.4 & \bf 34.4 & \bf 34.1 & \bf 35.1 \\
        pMKQS &    617 & \bf 292 &     146 &     84.2 & \bf 57.5 &     47.1 &     41.0 &     38.0 &     37.0 &     38.0 \\
     pRS-8bit & 1\,959 &     975 &     493 &      280 &      204 &      171 &      144 &      142 &      140 &      140 \\
    pRS-16bit & 1\,960 &     883 &     444 &      260 &      179 &      148 &      125 &      119 &      115 &      116 \\
    pRS/Akiba & 1\,293 &  1\,258 &  1\,255 &   1\,249 &   1\,256 &   1\,255 &   1\,249 &   1\,259 &   1\,255 &   1\,249 \\ \hline
& \multicolumn{10}{l|}{\textbf{Random}, $n = 3.27\,\text{G}$, $N = 32\,\text{Gi}$, $\frac{D}{N} = 44.9\,\%$} \\ \cline{2-11}
mkqs\_cache8 & \bf 609 &  &  &  &  &  &  &  &  &  \\
pS$^5$-Unroll & 1\,209 &     601 &     301 &     166 &      122 &      101 &     76.7 &     67.7 &     65.3 &     63.7 \\
 pS$^5$-Equal & 1\,322 &     657 &     326 &     178 &      131 &      107 &     81.4 &     70.7 &     67.6 &     64.2 \\
 pS$^5$+LCP-M &        &         &     367 &     196 &      135 &      108 &     80.0 &     71.0 &     72.7 &     71.7 \\
        pMKQS &    732 & \bf 379 & \bf 198 & \bf 108 & \bf 79.7 & \bf 69.1 & \bf 63.6 &     65.5 &     70.7 &     75.2 \\
     pRS-8bit & 1\,530 &     706 &     343 &     183 &      127 &      100 &     70.9 &     59.3 &     60.8 &     59.2 \\
    pRS-16bit & 1\,530 &     657 &     343 &     185 &      129 &      100 &     69.4 & \bf 56.2 & \bf 52.4 & \bf 53.6 \\
    pRS/Akiba & 1\,355 &     751 &     447 &     321 &      280 &      257 &      232 &      223 &      219 &      216 \\ \hline
& \multicolumn{10}{l|}{\textbf{GOV2}, $n = 3.1\,\text{G}$, $N = 128\,\text{Gi}$, $\frac{D}{N} = 82.7\,\%$} \\ \cline{2-11}
mkqs\_cache8 & \bf 1\,079 &  &  &  &  &  &  &  &  &  \\
pS$^5$-Unroll & 1\,399 &     673 &     326 &     212 &     176 &     145 &      113 &     97.1 &     92.1 &     89.0 \\
 pS$^5$-Equal & 1\,476 &     705 &     339 &     224 &     186 &     154 &      119 &      101 &     95.8 &     91.2 \\
 pS$^5$+LCP-M &        &         & \bf 272 & \bf 166 & \bf 131 & \bf 112 & \bf 93.3 & \bf 80.5 & \bf 80.7 & \bf 82.1 \\
        pMKQS & 1\,347 & \bf 661 &     350 &     207 &     164 &     127 &      101 &     91.8 &     93.7 &     95.1 \\
     pRS-8bit & 4\,244 &  1\,992 &     964 &     585 &     462 &     394 &      311 &      299 &      302 &      291 \\
    pRS-16bit & 4\,252 &  1\,912 &     928 &     571 &     471 &     384 &      306 &      279 &      280 &      257 \\
    pRS/Akiba & 2\,645 &  1\,306 &  1\,028 &  1\,055 &  1\,048 &  1\,034 &   1\,037 &   1\,045 &   1\,051 &   1\,052 \\ \hline
& \multicolumn{10}{l|}{\textbf{Wikipedia}, $n = N = 4\,\text{Gi}$, $D = 249\,\text{G}$} \\ \cline{2-11}
mkqs\_cache8 & \bf 2\,502 &  &  &  &  &  &  &  &  &  \\
pS$^5$-Unroll & 2\,728 &     1\,341 &     648 &     350 &     252 &     203 &     147 &     120 &      110 &      103 \\
 pS$^5$-Equal & 2\,986 &     1\,435 &     694 &     374 &     268 &     215 &     157 &     125 &      115 &      105 \\
 pS$^5$+LCP-M &        &            &     635 &     338 & \bf 235 & \bf 186 & \bf 130 & \bf 107 & \bf 97.6 & \bf 92.0 \\
        pMKQS & 2\,554 & \bf 1\,259 & \bf 620 & \bf 336 &     238 &     189 &     143 &     127 &      122 &      119 \\
     pRS-8bit & 4\,064 &     1\,879 &     909 &     486 &     349 &     271 &     192 &     157 &      150 &      144 \\
    pRS-16bit & 4\,068 &     1\,805 &     875 &     469 &     340 &     262 &     187 &     149 &      145 &      139 \\
    pRS/Akiba & 2\,862 &     1\,450 &     754 &     453 &     355 &     302 &     249 &     229 &      265 &      263 \\ \hline
& \multicolumn{10}{l|}{\textbf{Sinha NoDup} (complete), $n = 31.6\,\text{M}$, $N = 382\,\text{Mi}$, $\frac{D}{N} = 73.4\,\%$} \\ \cline{2-11}
radixR\_CE7 & \bf 6.00 &  &  &  &  &  &  &  &  &  \\
pS$^5$-Unroll & 8.27 &     4.17 &     2.13 &      1.22 &     0.921 &     0.790 &     0.695 &     0.642 &     0.592 &     0.544 \\
 pS$^5$-Equal & 8.84 &     4.46 &     2.26 &      1.28 &     0.963 &     0.825 &     0.710 &     0.661 &     0.601 &     0.567 \\
 pS$^5$+LCP-M &      &          &     2.52 &      1.41 &      1.01 &     0.815 &     0.653 &     0.604 &     0.691 &     0.779 \\
        pMKQS & 8.44 &     4.30 &     2.21 &      1.25 &     0.920 &     0.801 &     0.744 &     0.798 &     0.973 &      1.11 \\
     pRS-8bit & 8.35 &     4.06 &     2.01 &      1.08 &     0.770 &     0.643 &     0.489 &     0.425 & \bf 0.422 & \bf 0.441 \\
    pRS-16bit & 8.35 & \bf 3.49 & \bf 1.74 & \bf 0.949 & \bf 0.682 & \bf 0.595 & \bf 0.445 & \bf 0.422 &     0.467 &     0.502 \\
    pRS/Akiba & 7.58 &     4.09 &     2.38 &      1.59 &      1.33 &      1.21 &      1.09 &      1.04 &      1.04 &      1.02 \\ \hline
\end{tabular}
\end{table}


\begin{table}\centering\small
\caption{Absolute run time of parallel and best sequential algorithms on AMD48 in seconds, median of 1--3 runs. See Table~\ref{tab:paraalgo} for a short description of each.}\label{tab:absrun-AMD48}
\def\tabcolsep{2.8pt}
\begin{tabular}{l|*{11}{r}|@{}}
PEs & 1   & 2 & 3 & 6 & 9 & 12 & 18 & 24 & 36 & 42 & 48 \\ \hline
& \multicolumn{11}{l|}{\textbf{URLs} (complete), $n = 1.11\,\text{G}$, $N = 70.7\,\text{Gi}$, $\frac{D}{N} = 93.5\,\%$} \\ \cline{2-12}
mkqs\_cache8 & \bf 773 &  &  &  &  &  &  &  &  &  &  \\
pS$^5$-Unroll & 1\,030 &     521 &     352 &     181 &     127 &     99.9 &     74.9 &     64.0 &     53.0 &     48.0 &     46.8 \\
 pS$^5$-Equal &    931 &     477 &     331 &     176 &     123 &     97.1 &     73.3 &     65.0 &     55.0 &     48.8 &     47.6 \\
 pS$^5$+LCP-M &        &         &         &         &     122 &      114 &     76.3 &     60.3 &     50.0 &     49.6 &     46.8 \\
        pMKQS &    844 & \bf 415 & \bf 280 & \bf 146 & \bf 102 & \bf 80.4 & \bf 59.6 & \bf 49.9 & \bf 44.2 & \bf 44.0 & \bf 45.1 \\
     pRS-8bit & 2\,552 &  1\,306 &     897 &     468 &     325 &      266 &      202 &      177 &      157 &      156 &      155 \\
    pRS-16bit & 2\,550 &  1\,210 &     823 &     428 &     299 &      234 &      183 &      151 &      126 &      125 &      126 \\
    pRS/Akiba & 1\,861 &  1\,840 &  1\,832 &  1\,830 &  1\,821 &   1\,819 &   1\,823 &   1\,822 &   1\,822 &   1\,827 &   1\,821 \\ \hline
& \multicolumn{11}{l|}{\textbf{Random}, $n = 2.45\,\text{G}$, $N = 24\,\text{Gi}$, $\frac{D}{N} = 44.5\,\%$} \\ \cline{2-12}
mkqs\_cache8 & \bf 879 &  &  &  &  &  &  &  &  &  &  \\
pS$^5$-Unroll & 1\,315 &     683 &     466 &     248 &     176 &      140 &      104 &     86.3 &     69.3 &     64.4 &     61.7 \\
 pS$^5$-Equal & 1\,225 &     634 &     433 &     232 &     165 &      131 &     98.3 &     82.1 &     67.1 &     62.1 &     59.7 \\
 pS$^5$+LCP-M &        &         &         &         &     196 &      185 &      109 &     82.7 &     71.5 &     66.0 &     61.9 \\
        pMKQS &    751 & \bf 392 & \bf 264 & \bf 143 & \bf 106 & \bf 81.0 & \bf 65.3 & \bf 56.0 &     54.7 &     55.8 &     61.3 \\
     pRS-8bit & 1\,182 &     594 &     404 &     209 &     144 &      111 &     79.6 &     63.5 &     52.2 &     49.9 &     47.1 \\
    pRS-16bit & 1\,188 &     615 &     423 &     224 &     154 &      120 &     85.1 &     68.0 & \bf 51.3 & \bf 47.4 & \bf 44.8 \\
    pRS/Akiba & 1\,525 &     861 &     643 &     421 &     348 &      312 &      277 &      259 &      241 &      238 &      234 \\ \hline
& \multicolumn{11}{l|}{\textbf{GOV2}, $n = 1.38\,\text{G}$, $N = 64\,\text{Gi}$, $\frac{D}{N} = 77.0\,\%$} \\ \cline{2-12}
mkqs\_cache8 & \bf 750 &  &  &  &  &  &  &  &  &  &  \\
pS$^5$-Unroll &    881 &     449 &     305 &     162 &     129 &      110 &     83.2 &     69.9 &     58.4 &     54.4 &     50.2 \\
 pS$^5$-Equal &    854 &     436 &     296 &     156 &     123 &      104 &     79.2 &     67.1 &     55.8 &     52.5 &     48.7 \\
 pS$^5$+LCP-M &        &         &         &         & \bf 100 &     98.2 &     62.0 & \bf 48.7 & \bf 42.9 & \bf 39.5 & \bf 37.4 \\
        pMKQS &    785 & \bf 397 & \bf 268 & \bf 140 &     101 & \bf 83.8 & \bf 60.1 &     50.7 &     44.1 &     42.8 &     42.6 \\
     pRS-8bit & 2\,071 &  1\,015 &     676 &     351 &     280 &      251 &      182 &      154 &      125 &      120 &      114 \\
    pRS-16bit & 2\,061 &     980 &     652 &     338 &     269 &      242 &      171 &      143 &      111 &      107 &      102 \\
    pRS/Akiba & 1\,547 &     794 &     617 &     594 &     584 &      582 &      581 &      581 &      584 &      585 &      585 \\ \hline
& \multicolumn{11}{l|}{\textbf{Wikipedia}, $n = N = 2\,\text{Gi}$, $D = 116\,\text{G}$} \\ \cline{2-12}
mkqs\_cache8 & \bf 1\,442 &  &  &  &  &  &  &  &  &  &  \\
pS$^5$-Unroll & 1\,634 &     838 &     569 &     299 &     208 &     164 &     119 &     96.9 &     74.6 &     68.7 &     65.3 \\
 pS$^5$-Equal & 1\,592 &     827 &     561 &     293 &     205 &     161 &     117 &     95.4 &     73.6 &     68.1 &     64.3 \\
 pS$^5$+LCP-M &        &         &         &         &     267 &     255 &     142 &      104 &     83.4 &     72.8 &     65.9 \\
        pMKQS & 1\,556 & \bf 790 & \bf 534 & \bf 273 & \bf 188 & \bf 145 & \bf 102 & \bf 84.3 & \bf 67.0 & \bf 64.1 & \bf 64.0 \\
     pRS-8bit & 2\,547 &  1\,216 &     825 &     417 &     281 &     215 &     148 &      115 &     85.2 &     78.1 &     72.6 \\
    pRS-16bit & 2\,547 &  1\,168 &     795 &     405 &     276 &     211 &     147 &      114 &     82.6 &     74.2 &     69.2 \\
    pRS/Akiba & 1\,966 &  1\,030 &     717 &     403 &     299 &     249 &     198 &      197 &      196 &      196 &      198 \\ \hline
& \multicolumn{11}{l|}{\textbf{Sinha NoDup} (complete), $n = 31.6\,\text{M}$, $N = 382\,\text{Mi}$, $\frac{D}{N} = 73.4\,\%$} \\ \cline{2-12}
radixR\_CE7 & \bf 8.24 &  &  &  &  &  &  &  &  &  &  \\
pS$^5$-Unroll & 12.8 &     6.73 &     4.63 &     2.52 &     1.82 &      1.47 &      1.13 &     0.978 &     0.836 &     0.804 &     0.794 \\
 pS$^5$-Equal & 12.0 &     6.30 &     4.34 &     2.37 &     1.72 &      1.39 &      1.08 &     0.944 &     0.812 &     0.779 &     0.772 \\
 pS$^5$+LCP-M &      &          &          &          &     1.89 &      1.76 &      1.39 &      1.10 &      1.01 &     0.993 &     0.984 \\
        pMKQS & 11.2 &     5.81 &     4.02 &     2.14 &     1.53 &      1.28 &      1.00 &     0.935 &     0.989 &      1.04 &      1.16 \\
     pRS-8bit & 10.8 &     5.33 &     3.61 &     1.87 &     1.35 &      1.01 & \bf 0.771 & \bf 0.596 & \bf 0.482 & \bf 0.462 & \bf 0.453 \\
    pRS-16bit & 10.8 & \bf 4.82 & \bf 3.27 & \bf 1.72 & \bf 1.32 & \bf 0.988 &     0.872 &     0.779 &     0.924 &      1.01 &      1.10 \\
    pRS/Akiba & 11.5 &     6.33 &     4.62 &     2.90 &     2.33 &      2.05 &      1.80 &      1.68 &      1.57 &      1.54 &      1.53 \\ \hline
\end{tabular}
\end{table}


\begin{table}\centering\small
\caption{Absolute run time of parallel and best sequential algorithms on AMD16 in seconds, median of 1--3 runs. See Table~\ref{tab:paraalgo} for a short description of each.}\label{tab:absrun-AMD16}
\begin{tabularx}{\linewidth}{l|*{7}{>{\hfill}X}|@{}}
PEs & 1   & 2 & 4 & 6 & 8 & 12 & 16 \\ \hline
& \multicolumn{7}{l|}{\textbf{URLs}, $n = 500\,\text{M}$, $N = 32\,\text{Gi}$, $\frac{D}{N} = 95.4\,\%$} \\ \cline{2-8}
mkqs\_cache8 & \bf 422 &  &  &  &  &  &  \\
pS$^5$-Unroll &    424 & \bf 218 & \bf 115 & \bf 83.0 & \bf 70.5 & \bf 56.2 & \bf 49.1 \\
 pS$^5$-Equal &    455 &     233 &     123 &     91.4 &     75.1 &     61.6 &     51.6 \\
 pS$^5$+LCP-M &        &         &     152 &      152 &     91.9 &     73.2 &     58.7 \\
        pMKQS &    456 &     220 &     117 &     88.3 &     74.4 &     63.7 &     61.5 \\
     pRS-8bit & 1\,256 &     652 &     362 &      284 &      253 &      238 &      220 \\
    pRS-16bit & 1\,253 &     601 &     331 &      255 &      225 &      212 &      185 \\
    pRS/Akiba & 1\,063 &  1\,060 &  1\,049 &   1\,057 &   1\,110 &   1\,048 &   1\,105 \\ \hline
& \multicolumn{7}{l|}{\textbf{Random}, $n = 1.23\,\text{G}$, $N = 12\,\text{Gi}$, $\frac{D}{N} = 43.7\,\%$} \\ \cline{2-8}
mkqs\_cache8 & \bf 350 &  &  &  &  &  &  \\
pS$^5$-Unroll & 675 &     349 &     182 &      128 &      100 &     74.6 &     61.6 \\
 pS$^5$-Equal & 621 &     321 &     167 &      119 &     93.2 &     70.0 &     58.2 \\
 pS$^5$+LCP-M &     &         &     207 &      194 &      112 &     80.3 &     64.4 \\
        pMKQS & 384 & \bf 203 & \bf 112 & \bf 84.6 & \bf 79.6 &     69.4 &     62.4 \\
     pRS-8bit & 605 &     302 &     161 &      109 &     89.8 & \bf 67.6 & \bf 55.7 \\
    pRS-16bit & 605 &     297 &     167 &      110 &     95.5 &     71.6 &     59.1 \\
    pRS/Akiba & 805 &     459 &     283 &      227 &      198 &      171 &      157 \\ \hline
& \multicolumn{7}{l|}{\textbf{GOV2}, $n = 490\,\text{M}$, $N = 24\,\text{Gi}$, $\frac{D}{N} = 72.4\,\%$} \\ \cline{2-8}
mkqs\_cache8 & \bf 291 &  &  &  &  &  &  \\
pS$^5$-Unroll & 336 &     171 &     88.3 &     62.6 &     55.5 &     44.9 &     36.5 \\
 pS$^5$-Equal & 326 &     166 &     86.0 &     61.0 &     53.5 &     43.2 &     35.3 \\
 pS$^5$+LCP-M &     &         &     81.4 &     78.2 & \bf 46.5 & \bf 34.4 & \bf 28.0 \\
        pMKQS & 296 & \bf 152 & \bf 79.4 & \bf 60.6 &     46.9 &     39.4 &     37.2 \\
     pRS-8bit & 692 &     338 &      176 &      132 &      112 &      105 &     92.1 \\
    pRS-16bit & 691 &     327 &      170 &      126 &      115 &     92.9 &     81.7 \\
    pRS/Akiba & 552 &     285 &      226 &      225 &      225 &      226 &      227 \\ \hline
& \multicolumn{7}{l|}{\textbf{Wikipedia}, $n = N = 1\,\text{Gi}$, $D = 40\,\text{G}$} \\ \cline{2-8}
mkqs\_cache8 & \bf 642 &  &  &  &  &  &  \\
pS$^5$-Unroll & 840 &     424 &     214 &     147 &     112 &     78.9 &     62.0 \\
 pS$^5$-Equal & 819 &     414 &     209 &     144 &     110 &     77.1 &     60.6 \\
 pS$^5$+LCP-M &     &         &     195 &     175 &     104 & \bf 73.6 & \bf 57.7 \\
        pMKQS & 693 & \bf 351 & \bf 183 & \bf 130 & \bf 104 &     80.9 &     71.8 \\
     pRS-8bit & 920 &     425 &     216 &     153 &     118 &     90.1 &     75.1 \\
    pRS-16bit & 917 &     408 &     207 &     149 &     114 &     87.3 &     71.4 \\
    pRS/Akiba & 782 &     419 &     238 &     181 &     154 &      128 &      116 \\ \hline
& \multicolumn{7}{l|}{\textbf{Sinha NoDup} (complete), $n = 31.6\,\text{M}$, $N = 382\,\text{Mi}$, $\frac{D}{N} = 73.4\,\%$} \\ \cline{2-8}
radixR\_CE7 & \bf 9.50 &  &  &  &  &  &  \\
pS$^5$-Unroll & 11.9 &     6.13 &     3.15 &     2.14 &     1.65 &     1.19 &     0.963 \\
 pS$^5$-Equal & 11.0 &     5.65 &     2.90 &     1.99 & \bf 1.52 & \bf 1.11 & \bf 0.912 \\
 pS$^5$+LCP-M &      &          &     4.22 &     3.65 &     2.30 &     1.64 &      1.32 \\
        pMKQS & 11.9 &     6.11 &     3.23 &     2.34 &     1.92 &     1.56 &      1.46 \\
     pRS-8bit & 12.1 &     5.97 &     3.09 &     2.20 &     1.74 &     1.37 &      1.25 \\
    pRS-16bit & 12.1 & \bf 5.44 & \bf 2.82 & \bf 1.98 &     1.56 &     1.20 &      1.12 \\
    pRS/Akiba & 13.0 &     7.19 &     4.28 &     3.33 &     2.86 &     2.44 &      2.24 \\ \hline
\end{tabularx}
\end{table}


\begin{table}\centering\small
\caption{Absolute run time of parallel and best sequential algorithms on Inteli7 in seconds, median of fifteen runs, larger test instances. See Table~\ref{tab:paraalgo} for a short description of each.}\label{tab:absrun-Inteli7}
\begin{tabularx}{\linewidth}{l|*{8}{>{\hfill}X}|@{}}
PEs & 1   & 2 & 3 & 4 & 5 & 6 & 7 & 8 \\ \hline
& \multicolumn{8}{l|}{\textbf{URLs}, $n = 65.7\,\text{M}$, $N = 4\,\text{Gi}$, $\frac{D}{N} = 92.7\,\%$} \\ \cline{2-9}
mkqs\_cache8 & 14.8 &  &  &  &  &  &  &  \\
pS$^5$-Unroll &     14.9 &     8.46 &     5.87 & \bf 4.83 &     4.81 &     4.52 &     4.41 &     4.24 \\
 pS$^5$-Equal & \bf 14.5 & \bf 8.27 & \bf 5.80 &     4.86 & \bf 4.61 & \bf 4.39 & \bf 4.29 & \bf 4.17 \\
        pMKQS &     15.6 &     8.78 &     6.57 &     5.59 &     5.42 &     5.31 &     5.21 &     5.15 \\
     pRS-8bit &     39.8 &     22.0 &     16.8 &     14.5 &     14.3 &     13.9 &     13.6 &     13.5 \\
    pRS-16bit &     39.8 &     20.0 &     14.9 &     12.7 &     12.4 &     12.1 &     11.9 &     11.7 \\
    pRS/Akiba &     32.0 &     31.7 &     31.7 &     31.7 &     31.7 &     31.7 &     31.7 &     31.7 \\
     p2w-MS/R &     30.1 &     17.3 &     17.3 &     12.7 &     12.4 &     10.5 &     9.96 &     9.83 \\
 pMKQS-SIMD/R &     31.4 &     19.4 &     15.4 &     13.6 &     13.4 &     13.3 &     13.3 &     13.3 \\
 pLCP-2w-MS/S &     29.7 &     18.9 &     14.7 &     13.2 &     13.0 &     13.2 &     12.9 &     13.0 \\ \hline
& \multicolumn{8}{l|}{\textbf{Random}, $n = 205\,\text{M}$, $N = 2\,\text{Gi}$, $\frac{D}{N} = 42.1\,\%$} \\ \cline{2-9}
radixR\_CE7 & \bf 19.6 &  &  &  &  &  &  &  \\
pS$^5$-Unroll & 32.2 &     17.5 &     12.1 &     9.39 &     9.51 &     9.07 &     8.29 &     7.67 \\
 pS$^5$-Equal & 34.2 &     18.6 &     12.8 &     9.93 &     9.76 &     9.27 &     8.46 &     7.83 \\
        pMKQS & 31.2 &     16.9 &     11.9 &     9.42 &     8.64 &     8.07 &     7.75 &     7.50 \\
     pRS-8bit & 23.6 &     12.9 &     9.24 & \bf 7.89 & \bf 7.55 & \bf 7.22 & \bf 7.00 & \bf 6.86 \\
    pRS-16bit & 23.6 & \bf 12.5 & \bf 9.04 &     8.35 &     7.76 &     7.65 &     7.16 &     7.34 \\
    pRS/Akiba & 24.0 &     16.1 &     13.3 &     12.0 &     11.7 &     11.4 &     11.2 &     11.0 \\
     p2w-MS/R & 95.9 &     58.1 &     58.1 &     52.4 &     46.9 &     43.5 &     40.2 &     40.2 \\
 pMKQS-SIMD/R &  123 &     72.9 &     55.3 &     47.2 &     46.8 &     46.6 &     46.5 &     46.2 \\
 pLCP-2w-MS/S &  185 &      119 &      101 &     97.9 &      105 &      110 &      115 &      121 \\ \hline
& \multicolumn{8}{l|}{\textbf{GOV2}, $n = 80\,\text{M}$, $N = 4\,\text{Gi}$, $\frac{D}{N} = 69.8\,\%$} \\ \cline{2-9}
mkqs\_cache8 & 14.6 &  &  &  &  &  &  &  \\
pS$^5$-Unroll & \bf 13.1 & \bf 7.15 & \bf 5.00 & \bf 3.95 &     3.99 &     3.93 &     3.73 &     3.51 \\
 pS$^5$-Equal &     13.5 &     7.39 &     5.16 &     4.06 & \bf 3.94 & \bf 3.87 & \bf 3.69 & \bf 3.49 \\
        pMKQS &     15.5 &     8.62 &     6.22 &     5.08 &     4.82 &     4.58 &     4.43 &     4.30 \\
     pRS-8bit &     33.5 &     17.8 &     12.9 &     10.2 &     9.85 &     9.52 &     9.43 &     9.69 \\
    pRS-16bit &     33.4 &     17.1 &     12.2 &     9.62 &     9.20 &     8.73 &     8.40 &     8.20 \\
    pRS/Akiba &     27.4 &     15.2 &     13.4 &     13.3 &     13.3 &     14.2 &     14.4 &     14.4 \\
     p2w-MS/R &     32.6 &     19.3 &     19.3 &     16.8 &     15.8 &     12.8 &     12.6 &     12.4 \\
 pMKQS-SIMD/R &     32.5 &     19.5 &     14.9 &     12.8 &     12.6 &     12.5 &     12.4 &     12.4 \\
 pLCP-2w-MS/S &     43.9 &     22.9 &     15.8 &     12.5 &     12.8 &     11.5 &     10.5 &     9.94 \\ \hline
& \multicolumn{8}{l|}{\textbf{Wikipedia}, $n = N = 256\,\text{Mi}$, $D = 13.8\,\text{G}$} \\ \cline{2-9}
radixR\_CE7 & \bf 59.6 &  &  &  &  &  &  &  \\
pS$^5$-Unroll & 59.9 & \bf 32.8 & \bf 22.6 & \bf 17.5 & \bf 17.1 & \bf 16.1 & \bf 14.7 & \bf 13.6 \\
 pS$^5$-Equal & 63.3 &     34.6 &     23.8 &     18.4 &     17.5 &     16.3 &     15.0 &     13.9 \\
        pMKQS & 69.3 &     37.5 &     26.1 &     20.4 &     18.6 &     17.3 &     16.2 &     15.5 \\
     pRS-8bit & 74.5 &     37.9 &     26.2 &     19.9 &     18.6 &     17.2 &     16.1 &     15.1 \\
    pRS-16bit & 74.5 &     36.2 &     25.2 &     19.0 &     17.9 &     16.5 &     15.4 &     14.5 \\
    pRS/Akiba & 62.4 &     34.8 &     24.9 &     20.0 &     18.6 &     17.5 &     16.6 &     15.7 \\
     p2w-MS/R &  123 &     72.2 &     72.7 &     64.8 &     58.6 &     51.6 &     48.0 &     47.7 \\
 pMKQS-SIMD/R &  127 &     74.4 &     55.5 &     46.8 &     45.7 &     44.7 &     44.2 &     43.5 \\
 pLCP-2w-MS/S &  221 &      115 &     79.2 &     62.1 &     60.6 &     54.4 &     51.5 &     48.4 \\ \hline
\end{tabularx}
\end{table}

\def\tabcolsep{4pt}
\begin{table}\centering\small
\caption{Absolute run time of parallel and best sequential algorithms on Inteli7 in seconds, median of fifteen runs, smaller test instances. See Table~\ref{tab:paraalgo} for a short description of each.}\label{tab:absrun-Inteli7b}
\begin{tabularx}{\linewidth}{l|*{8}{>{\hfill}X}|@{}}
PEs          & 1   & 2 & 3 & 4 & 5 & 6 & 7 & 8                                                                                      \\ \hline
& \multicolumn{8}{l|}{\textbf{Sinha URLs} (complete), $n = 10\,\text{M}$, $N = 304\,\text{Mi}$, $\frac{D}{N} = 97.5\,\%$} \\ \cline{2-9}
mkqs\_cache8 & 1.81 &  &  &  &  &  &  &  \\
pS$^5$-Unroll & \bf 1.54 & \bf 0.853 & \bf 0.595 & \bf 0.471 & \bf 0.495 & \bf 0.464 & \bf 0.445 & \bf 0.431 \\
 pS$^5$-Equal &     1.67 &     0.924 &     0.641 &     0.505 &     0.505 &     0.475 &     0.452 &     0.436 \\
        pMKQS &     1.93 &      1.07 &     0.766 &     0.618 &     0.593 &     0.571 &     0.552 &     0.544 \\
     pRS-8bit &     5.05 &      2.48 &      1.74 &      1.37 &      1.32 &      1.28 &      1.23 &      1.20 \\
    pRS-16bit &     5.06 &      2.23 &      1.56 &      1.21 &      1.19 &      1.14 &      1.11 &      1.08 \\
    pRS/Akiba &     4.02 &      3.44 &      3.23 &      3.13 &      3.12 &      3.10 &      3.10 &      3.08 \\
     p2w-MS/R &     3.98 &      2.35 &      2.34 &      2.05 &      1.86 &      1.64 &      1.51 &      1.53 \\
 pMKQS-SIMD/R &     3.93 &      2.37 &      1.83 &      1.59 &      1.60 &      1.59 &      1.61 &      1.61 \\
 pLCP-2w-MS/S &     5.93 &      3.57 &      2.89 &      2.65 &      2.83 &      2.81 &      2.86 &      2.93 \\ \hline
& \multicolumn{8}{l|}{\textbf{Sinha DNA} (complete), $n = 31.6\,\text{M}$, $N = 302\,\text{Mi}$, $\frac{D}{N} = 100\,\%$} \\ \cline{2-9}
radixR\_CE6 & 3.69 &  &  &  &  &  &  &  \\
pS$^5$-Unroll & \bf 2.94 & \bf 1.63 & \bf 1.14 & \bf 0.906 & \bf 0.989 & \bf 0.900 & \bf 0.883 & \bf 0.831 \\
 pS$^5$-Equal &     3.37 &     1.85 &     1.29 &      1.02 &      1.04 &     0.944 &     0.920 &     0.864 \\
        pMKQS &     4.28 &     2.37 &     1.72 &      1.40 &      1.32 &      1.25 &      1.22 &      1.21 \\
     pRS-8bit &     5.47 &     2.99 &     2.14 &      1.68 &      1.63 &      1.58 &      1.54 &      1.51 \\
    pRS-16bit &     5.47 &     2.78 &     1.94 &      1.57 &      1.54 &      1.50 &      1.47 &      1.42 \\
    pRS/Akiba &     3.83 &     2.29 &     1.74 &      1.50 &      1.51 &      1.54 &      1.50 &      1.43 \\
     p2w-MS/R &     11.1 &     6.36 &     6.36 &      5.13 &      5.10 &      4.17 &      3.98 &      3.97 \\
 pMKQS-SIMD/R &     10.5 &     6.50 &     5.10 &      4.51 &      4.51 &      4.53 &      4.55 &      4.56 \\
 pLCP-2w-MS/S &     18.5 &     11.1 &     8.79 &      7.91 &      8.39 &      8.18 &      8.29 &      8.40 \\ \hline
& \multicolumn{8}{l|}{\textbf{Sinha NoDup} (complete), $n = 31.6\,\text{M}$, $N = 382\,\text{Mi}$, $\frac{D}{N} = 73.4\,\%$} \\ \cline{2-9}
radixR\_CE7 & \bf 3.83 &  &  &  &  &  &  &  \\
pS$^5$-Unroll & 4.54 & \bf 2.47 & \bf 1.69 & \bf 1.31 &     1.33 &     1.21 &     1.15 &     1.07 \\
 pS$^5$-Equal & 4.96 &     2.68 &     1.83 &     1.42 &     1.38 &     1.29 &     1.18 &     1.10 \\
        pMKQS & 5.47 &     2.98 &     2.11 &     1.67 &     1.52 &     1.42 &     1.34 &     1.29 \\
     pRS-8bit & 5.22 &     2.71 &     1.87 &     1.47 &     1.40 &     1.31 &     1.25 &     1.20 \\
    pRS-16bit & 5.22 &     2.48 &     1.70 &     1.32 & \bf 1.28 & \bf 1.19 & \bf 1.12 & \bf 1.07 \\
    pRS/Akiba & 4.97 &     2.84 &     2.06 &     1.68 &     1.59 &     1.50 &     1.43 &     1.37 \\
     p2w-MS/R & 12.7 &     7.46 &     7.45 &     6.22 &     5.96 &     4.89 &     4.88 &     4.89 \\
 pMKQS-SIMD/R & 13.6 &     8.08 &     6.10 &     5.23 &     5.15 &     5.09 &     5.05 &     5.02 \\
 pLCP-2w-MS/S & 19.8 &     12.2 &     10.0 &     9.39 &     9.43 &     10.1 &     10.4 &     10.7 \\ \hline
\end{tabularx}
\end{table}


\begin{table}\centering\small
\caption{Absolute run time of parallel and best sequential algorithms on IntelX5 in seconds, median of fifteen runs, larger test instances. See Table~\ref{tab:paraalgo} for a short description of each.}\label{tab:absrun-IntelX5}
\begin{tabularx}{\linewidth}{l|*{8}{>{\hfill}X}|@{}}
PEs          & 1   & 2 & 3 & 4 & 5 & 6 & 7 & 8                                                                                      \\ \hline
             & \multicolumn{8}{l|}{\textbf{URLs}, $n = 132\,\text{M}$, $N = 8\,\text{Gi}$, $\frac{D}{N} = 92.6\,\%$} \\ \cline{2-9}
mkqs\_cache8 & 64.2 &  &  &  &  &  &  &  \\
pS$^5$-Unroll & \bf 56.2 & \bf 30.8 & \bf 24.4 & \bf 20.7 & \bf 20.2 & \bf 19.2 & \bf 19.0 & \bf 19.0 \\
 pS$^5$-Equal &     58.6 &     32.1 &     25.3 &     21.2 &     20.5 &     19.7 &     19.4 &     19.3 \\
        pMKQS &     67.1 &     36.0 &     29.8 &     26.5 &     26.2 &     26.0 &     26.0 &     25.9 \\
     pRS-8bit &      150 &     89.5 &     79.6 &     72.9 &     72.8 &     71.7 &     72.6 &     71.6 \\
    pRS-16bit &      150 &     78.5 &     68.3 &     62.7 &     61.9 &     61.8 &     62.4 &     59.7 \\
    pRS/Akiba &      121 &      119 &      120 &      119 &      120 &      119 &      120 &      119 \\
     p2w-MS/R &     85.9 &     49.3 &     52.5 &     42.1 &     40.1 &     35.0 &     35.1 &     34.9 \\
 pMKQS-SIMD/R &      153 &     92.0 &     83.9 &     77.9 &     77.5 &     77.2 &     77.4 &     77.5 \\
 pLCP-2w-MS/S &      108 &     63.3 &     52.5 &     49.9 &     57.8 &     55.8 &     54.4 &     57.7 \\ \hline
& \multicolumn{8}{l|}{\textbf{Random}, $n = 307\,\text{M}$, $N = 3\,\text{Gi}$, $\frac{D}{N} = 42.8\,\%$} \\ \cline{2-9}
mkqs\_cache8 & \bf 58.9 &  &  &  &  &  &  &  \\
pS$^5$-Unroll & 78.5 &     42.2 &     31.0 &     25.0 & \bf 22.9 & \bf 20.9 & \bf 19.7 & \bf 18.9 \\
 pS$^5$-Equal & 82.1 &     44.0 &     32.1 &     25.9 &     23.6 &     21.1 &     20.0 &     19.2 \\
        pMKQS & 65.8 & \bf 35.5 & \bf 28.4 & \bf 24.4 &     23.7 &     23.5 &     23.4 &     23.5 \\
     pRS-8bit & 77.5 &     44.1 &     35.4 &     30.4 &     28.4 &     26.9 &     26.0 &     25.5 \\
    pRS-16bit & 77.5 &     43.3 &     35.0 &     27.9 &     27.0 &     25.5 &     24.3 &     23.8 \\
    pRS/Akiba & 81.7 &     50.6 &     41.9 &     37.2 &     34.9 &     33.3 &     32.6 &     31.9 \\
     p2w-MS/R &  303 &      186 &      200 &      172 &      168 &      156 &      154 &      155 \\
 pMKQS-SIMD/R &  467 &      277 &      246 &      229 &      225 &      222 &      220 &      219 \\
 pLCP-2w-MS/S &  625 &      388 &      390 &      377 &      386 &      412 &      472 &      507 \\ \hline
& \multicolumn{8}{l|}{\textbf{GOV2}, $n = 166\,\text{M}$, $N = 8\,\text{Gi}$, $\frac{D}{N} = 70.6\,\%$} \\ \cline{2-9}
mkqs\_cache8 & 55.3 &  &  &  &  &  &  &  \\
pS$^5$-Unroll & \bf 47.8 & \bf 25.8 & \bf 19.5 & \bf 16.2 & \bf 15.4 & \bf 14.5 & \bf 14.8 & \bf 14.6 \\
 pS$^5$-Equal &     49.1 &     26.4 &     19.9 &     16.4 &     15.5 &     14.6 &     15.0 &     14.8 \\
        pMKQS &     58.7 &     32.2 &     26.0 &     22.6 &     22.0 &     22.0 &     22.1 &     22.2 \\
     pRS-8bit &      115 &     61.7 &     49.8 &     43.8 &     44.9 &     45.2 &     44.2 &     44.0 \\
    pRS-16bit &      115 &     59.5 &     47.0 &     40.0 &     37.7 &     40.4 &     38.8 &     38.8 \\
    pRS/Akiba &     94.2 &     51.6 &     49.3 &     48.3 &     49.2 &     50.6 &     52.6 &     53.0 \\
     p2w-MS/R &      124 &     74.7 &     81.7 &     69.0 &     66.2 &     62.0 &     62.1 &     62.7 \\
 pMKQS-SIMD/R &      165 &     98.4 &     87.8 &     81.1 &     79.6 &     79.2 &     78.6 &     79.3 \\
 pLCP-2w-MS/S &      185 &     93.9 &     90.1 &     71.2 &     67.0 &     62.4 &     66.0 &     65.3 \\ \hline
& \multicolumn{8}{l|}{\textbf{Wikipedia}, $n = N = 512\,\text{Mi}$, $D = 21.5\,\text{G}$} \\ \cline{2-9}
radixR\_CE7 & \bf 185 &  &  &  &  &  &  &  \\
pS$^5$-Unroll & 194 & \bf 104 & \bf 76.2 & \bf 62.0 & \bf 56.8 & \bf 52.5 & \bf 50.8 & \bf 49.4 \\
 pS$^5$-Equal & 202 &     107 &     78.7 &     63.8 &     58.2 &     53.6 &     51.7 &     50.3 \\
        pMKQS & 211 &     112 &     87.5 &     74.2 &     71.3 &     69.4 &     68.9 &     68.9 \\
     pRS-8bit & 212 &     111 &     85.4 &     71.0 &     67.0 &     63.8 &     62.6 &     61.7 \\
    pRS-16bit & 212 &     109 &     84.2 &     68.5 &     64.8 &     62.1 &     59.8 &     58.2 \\
    pRS/Akiba & 192 &     107 &     83.4 &     70.6 &     66.5 &     64.3 &     62.9 &     62.5 \\
     p2w-MS/R & 420 &     261 &      281 &      237 &      231 &      230 &      230 &      230 \\
 pMKQS-SIMD/R & 583 &     343 &      301 &      277 &      274 &      273 &      273 &      274 \\
 pLCP-2w-MS/S & 882 &     443 &      345 &      277 &      292 &      282 &      252 &      264 \\ \hline
\end{tabularx}
\end{table}

\begin{table}\centering\small
\caption{Absolute run time of parallel and best sequential algorithms on IntelX5 in seconds, median of fifteen runs, smaller test instances. See Table~\ref{tab:paraalgo} for a short description of each.}\label{tab:absrun-IntelX5b}
\begin{tabularx}{\linewidth}{l|*{8}{>{\hfill}X}|@{}}
PEs          & 1   & 2 & 3 & 4 & 5 & 6 & 7 & 8                                                                                      \\ \hline
& \multicolumn{8}{l|}{\textbf{Sinha URLs} (complete), $n = 10\,\text{M}$, $N = 304\,\text{Mi}$, $\frac{D}{N} = 97\,5\,\%$} \\ \cline{2-9}
mkqs\_cache8 & 3.35 &  &  &  &  &  &  &  \\
pS$^5$-Unroll & \bf 2.91 & \bf 1.60 & \bf 1.22 & \bf 1.03 & \bf 0.985 & \bf 0.939 &     0.932 &     0.929 \\
 pS$^5$-Equal &     3.03 &     1.65 &     1.26 &     1.05 &     0.986 &     0.940 & \bf 0.931 & \bf 0.929 \\
        pMKQS &     3.57 &     1.95 &     1.57 &     1.35 &      1.33 &      1.31 &      1.32 &      1.34 \\
     pRS-8bit &     5.85 &     3.22 &     2.63 &     2.30 &      2.22 &      2.18 &      2.16 &      2.14 \\
    pRS-16bit &     5.86 &     2.86 &     2.32 &     1.99 &      1.93 &      1.90 &      1.87 &      1.86 \\
    pRS/Akiba &     5.12 &     4.38 &     4.22 &     4.10 &      4.09 &      4.10 &      4.09 &      4.08 \\
     p2w-MS/R &     6.40 &     4.02 &     4.31 &     3.67 &      3.55 &      3.45 &      3.55 &      3.49 \\
 pMKQS-SIMD/R &     9.56 &     5.87 &     5.47 &     5.15 &      5.15 &      5.14 &      5.20 &      5.20 \\
 pLCP-2w-MS/S &     10.2 &     6.10 &     5.44 &     5.12 &      5.89 &      5.84 &      5.98 &      6.27 \\ \hline
& \multicolumn{8}{l|}{\textbf{Sinha DNA} (complete), $n = 31.6\,\text{M}$, $N = 302\,\text{Mi}$, $\frac{D}{N} = 100\,\%$} \\ \cline{2-9}
radixR\_CE7 & 6.11 &  &  &  &  &  &  &  \\
pS$^5$-Unroll & \bf 5.14 & \bf 2.87 & \bf 2.21 & \bf 1.85 & \bf 1.73 & \bf 1.59 & \bf 1.55 & \bf 1.50 \\
 pS$^5$-Equal &     5.63 &     3.11 &     2.36 &     1.96 &     1.82 &     1.66 &     1.60 &     1.54 \\
        pMKQS &     7.14 &     3.97 &     3.39 &     2.98 &     2.95 &     2.96 &     2.96 &     2.98 \\
     pRS-8bit &     7.45 &     4.36 &     3.85 &     3.37 &     3.36 &     3.33 &     3.33 &     3.36 \\
    pRS-16bit &     7.45 &     4.05 &     3.46 &     3.10 &     3.09 &     3.07 &     3.08 &     3.07 \\
    pRS/Akiba &     6.00 &     3.72 &     3.44 &     3.12 &     3.13 &     3.11 &     3.12 &     3.06 \\
     p2w-MS/R &     18.1 &     10.6 &     11.3 &     8.96 &     8.91 &     7.92 &     7.98 &     7.92 \\
 pMKQS-SIMD/R &     27.1 &     16.8 &     15.6 &     14.7 &     14.7 &     14.7 &     14.8 &     14.8 \\
 pLCP-2w-MS/S &     33.1 &     20.2 &     17.8 &     16.5 &     18.6 &     18.4 &     18.6 &     19.4 \\ \hline
& \multicolumn{8}{l|}{\textbf{Sinha NoDup} (complete), $n = 31.6\,\text{M}$, $N = 382\,\text{Mi}$, $\frac{D}{N} = 73.4\,\%$} \\ \cline{2-9}
radixR\_CE7 & \bf 5.96 &  &  &  &  &  &  &  \\
pS$^5$-Unroll & 7.06 &     3.83 &     2.80 & \bf 2.25 & \bf 2.02 & \bf 1.88 & \bf 1.74 & \bf 1.66 \\
 pS$^5$-Equal & 7.51 &     4.05 &     2.94 &     2.35 &     2.08 &     1.88 &     1.77 &     1.69 \\
        pMKQS & 8.17 &     4.39 &     3.54 &     2.97 &     2.85 &     2.82 &     2.76 &     2.76 \\
     pRS-8bit & 7.22 &     3.99 &     3.20 &     2.74 &     2.62 &     2.49 &     2.46 &     2.43 \\
    pRS-16bit & 7.22 & \bf 3.57 & \bf 2.78 &     2.32 &     2.19 &     2.10 &     2.03 &     1.97 \\
    pRS/Akiba & 7.10 &     4.10 &     3.27 &     2.81 &     2.66 &     2.55 &     2.51 &     2.48 \\
     p2w-MS/R & 20.8 &     12.6 &     13.4 &     11.3 &     11.3 &     10.4 &     10.2 &     10.1 \\
 pMKQS-SIMD/R & 30.0 &     17.9 &     15.9 &     14.8 &     14.7 &     14.6 &     14.6 &     14.7 \\
 pLCP-2w-MS/S & 34.5 &     21.3 &     19.5 &     18.9 &     21.7 &     22.0 &     22.5 &     24.5 \\ \hline
\end{tabularx}
\end{table}

\begin{table}\centering\small
\caption{Description of parallel string sorting algorithms in experiment}\label{tab:paraalgo}
\begin{tabularx}{\linewidth}{l|X}
Name          & Description and Author                                                                                                                                                            \\ \hline
pS$^5$-Unroll & Our parallel super scalar string sample sort (see Section~\ref{sec:s5}) with unrolled and interleaved tree descents.                                                              \\
 pS$^5$-Equal & Our parallel super scalar string sample sort (see Section~\ref{sec:s5}) with equality checking at each splitter node.                                                             \\
        pMKQS & Our parallel multikey quicksort (see Section~\ref{sec:para-mkqs}) with caching of $w = 8$ characters.                                                                             \\
pS$^5$+LCP-M  & Our parallel multiway LCP-merge with pS$^5$ on each NUMA node.                                                                                                                \\
     pRS-8bit & Our parallel radix sort (see Section~\ref{sec:para-radixsort}) with $8$-bit alphabet                                                                                              \\
    pRS-16bit & Our parallel radix sort (see Section~\ref{sec:para-radixsort}) with $16$-bit alphabet at the fully parallel levels, and $8$-bit alphabets for sequentially processed subproblems. \\
    pRS/Akiba & Ta\-kuya Akiba's~\cite{akiba2011radixsort} radix sort.                                                                                                                            \\
     p2w-MS/R & Parallel 2-way LCP-mergesort from Tommi Rantala's  library~\cite{rantala2007web}.                                                                                                 \\
 pMKQS-SIMD/R & Parallel multikey quicksort with SIMD operations from Tommi Rantala's library~\cite{rantala2007web}.                                                                              \\
 pLCP-2w-MS/S & Parallel 2-way LCP-mergesort by Nagaraja Shamsundar~\cite{shamsundar2009lcpmergesort}, which is based on Waihong Ng's LCP-mergesort~\cite{ng2008merging}.                         \\ \hline
\end{tabularx}
\end{table}

\FloatBarrier

\section{Conclusions and Future Work}\label{sec:conclusions}

We have demonstrated that string sorting can be parallelized successfully on
modern multi-core shared memory and NUMA machines. In particular, our new string
sample sort algorithm combines favorable features of some of the best sequential
algorithms -- robust multiway divide-and-conquer from burstsort, efficient data
distribution from radix sort, asymptotic guarantees similar to multikey
quicksort, and word parallelism from caching multikey quicksort.  For NUMA
machines we developed parallel $K$-way LCP-merge to further decrease costly
inter-node random access.

We want to highlight that using our pS$^5$ (which can save LCPs) and $K$-way
LCP-merge implementations it is straight-forward to construct a fast parallel
external memory string sorter for short strings ($\leq B$) using shared memory
parallelism. The sorting throughput of our string sorters is probably higher
than the available I/O bandwidth.

Implementing some of the refinements discussed in the next section are likely to
yield further improvements for string sample sort and $K$-way
LCP-merge.

\subsection{Practical Refinements}\label{app:refinements}

\emph{Memory conservation:} For use of our algorithms in applications like
database systems or MapReduce libraries, it is paramount to give hard guarantees
on the amount of memory required by the implementations. Our experiments show
clearly, that caching of characters accelerates string sorting, but this speed
comes at the cost of memory space. A future challenge is thus to sort fast, but
with limited memory. In this respect, pS$^5$ is a very promising candidate, as
it can be restricted to use only the classification tree and a recursion stack,
if little additional memory is available. But if more memory is available, then
caching, saving oracle values and out-of-place redistribution can be enabled
adaptively.

\emph{Multipass data distribution:} There are two constraints for the maximum
sensible value for the number of splitters $v$: The cache size needed for the
classification data structure and the resources needed for data
distribution. Already in the plain external memory model these two constraints
differ ($v=\Oh{M}$ versus $v=\Oh{M/B}$). In practice, things are even more
complicated since multiple cache levels, cache replacement policy, TLBs,
etc. play a role. Anyway, we can increase the value of $v$ to the value required
for classification by doing the data distribution in multiple passes (usually
two).  Note that this fits very well with our approach to compute oracles even
for single pass data distribution. This approach can be viewed as LSD radix sort
using the oracles as keys. Initial experiments indicate that this could indeed
lead to some performance improvements.

\emph{Alphabet compression:} When we know that only $\sigma'<\sigma$ different
values from $\Sigma$ appear in the input, we can compress characters into
$\ceil{\log\sigma'}$ bits.  For S$^5$, this allows us to pack more characters
into a single machine word.  For example, for DNA input, we might pack 32
characters into a single 64 bit machine word. Note that this compression can be
done on the fly without changing the input/output format and the compression
overhead is amortized over $\log v$ key comparisons.

\emph{Jump tables:} In S$^5$, the $a$ most significant bits of a key are often
already sufficient to define a path in the classification tree of length up to
$a$.  We can exploit this by precomputing a jump table of size $2^a$ storing a
pointer to the end of this path. During element classification, a lookup in this
jump table can replace the traversal of the path.  This might reduce the gap to
radix sort for easy instances.

\emph{Using tries in practice:} The success of burstsort indicates that
traversing tries can be made efficient. Thus, we might also be able to use a
tuned trie based implementation of S$^5$ in practice. One ingredient to such an
implementation could be the word parallelism used in the pragmatic solution --
we define the trie over an enlarged alphabet. This reduces the number of
required hash table accesses by a factor of $w$. The tuned van Emde Boas trees
from \cite{dementiev2004sortedlist} suggest that this data structure might work
in practice.

\emph{Adaptivity:} By inspecting the sample, we can adaptively tune the
algorithm.  For example, when noticing that already a lot of information%
\footnote{The entropy $\frac{1}{n}\sum_i\log\frac{n}{|b_i|}$ can be used to
  define the amount of information gained by a set of splitters. The bucket
  sizes $b_i$ can be estimated using their size within the sample.} can be
gained from a few most significant bits in the sample keys, the algorithm might
decide to switch to radix sort. On the other hand, when even the $w$ most
significant characters do not give a lot of information, then a trie based
implementation can be used.  Again, this trie can be adapted to the input, for
example, using hash tables for low degree trie nodes and arrays for high degree
nodes.


\small
\bibliographystyle{spmpsci}
\bibliography{library}


\end{document}